\PassOptionsToPackage{dvipsnames}{xcolor}
\PassOptionsToPackage{nosort,nocompress}{cite}
\documentclass[a4paper,colorlinks]{preprint}
\usepackage[english]{babel}
\usepackage{mathtools}
\numberwithin{equation}{section}
\usepackage{amsthm}
\usepackage{thmtools}
\usepackage{amssymb}
\usepackage{bm}
\usepackage{amsfonts}
\usepackage{enumitem}
\setlist[itemize]{noitemsep, topsep=3pt}
\usepackage{microtype}
\usepackage{algorithm}
\usepackage{algpseudocode}
\usepackage{float}
\usepackage{comment}
\usepackage[nameinlink,capitalize]{cleveref}

\AddToHook{cmd/oldmaketitle/before}{%
  \ihead*{}\chead*{}\ohead*{}
  \ifoot*{}\ofoot*{}
  \cfoot*{\usekomafont{pagenumber}\thepage}
  \KOMAoptions{headsepline=false,plainheadsepline=false,
    footsepline=false,plainfootsepline=false}
}

\newtheorem{theorem}{Theorem}[section]
\newtheorem{lemma}[theorem]{Lemma}
\newtheorem{corollary}[theorem]{Corollary}
\newtheorem{proposition}[theorem]{Proposition}
\newtheorem{fact}[theorem]{Fact}
\newtheorem{problem}[theorem]{Problem}
\newtheorem{definition}[theorem]{Definition}
\newtheorem{observation}[theorem]{Observation}
\newtheorem{remark}[theorem]{Remark}

\crefname{appendix}{Appendix}{Appendices}
\Crefname{appendix}{Appendix}{Appendices}
\crefname{remark}{remark}{remarks}
\Crefname{remark}{Remark}{Remarks}
\crefname{observation}{observation}{observations}
\Crefname{observation}{Observation}{Observations}
\crefname{proposition}{proposition}{propositions}
\Crefname{proposition}{Proposition}{Propositions}
\crefname{fact}{fact}{facts}
\Crefname{fact}{Fact}{Facts}
\crefname{problem}{problem}{problems}
\Crefname{problem}{Problem}{Problems}
\crefname{algorithm}{algorithm}{algorithms}
\Crefname{algorithm}{Algorithm}{Algorithms}

\newcommand{\Nset}{\mathbb{N}}
\newcommand{\Qset}{\mathbb{Q}}
\newcommand{\Cset}{\mathbb{C}}
\newcommand{\poly}{\operatorname{poly}}

\newcommand{\Z}{\mathbb{Z}}

\newcommand{\bx}{\mathbf{x}}
\newcommand{\bz}{\mathbf{z}}
\newcommand{\ba}{\mathbf{a}}
\newcommand{\bb}{\mathbf{b}}
\newcommand{\bc}{\mathbf{c}}
\newcommand{\be}{\mathbf{e}}

\newcommand{\bq}{\mathbf{q}}
\newcommand{\br}{\mathbf{r}}
\newcommand{\bu}{\mathbf{u}}
\newcommand{\bv}{\mathbf{v}}
\newcommand{\bw}{\mathbf{w}}
\newcommand{\bE}{\mathbf{E}}
\newcommand{\balpha}{\bm{\alpha}}
\newcommand{\bbeta}{\bm{\beta}}
\newcommand{\bgamma}{\bm{\gamma}}
\newcommand{\supp}{\mathsf{supp}}

\newcommand{\rank}{\operatorname{rank}}
\newcommand{\vspan}{\operatorname{span}}

\renewcommand{\sp}{\operatorname{sp}}
\newcommand{\Ht}{\mathrm{Ht}}
\newcommand{\paren}[1]{\left( #1 \right)}

\title{Deterministic Shift Equivalence Testing for Sparse Polynomials}
\author[1]{Markus~Bl\"aser}
\author[2]{Pranjal~Dutta}
\author[3]{Sagnik~Dutta}
\author[4]{Gorav~Jindal\thanks{Funded by the European Union (ERC, CountHom, 101077083). Views and opinions expressed are however those of the author(s) only and do not necessarily reflect those of the European Union or the European Research Council Executive Agency.}}
\affil[1]{Saarland University, Saarland Informatics Campus, Saarbr\"ucken, Germany\authorcr
  \email{mblaeser@cs.uni-saarland.de}}
\affil[2]{Nanyang Technological University, Singapore\authorcr
  \email{pranjal.dutta@ntu.edu.sg}}
\affil[3]{Max Planck Institute for Informatics, Saarland Informatics Campus, Saarbr\"ucken, Germany\authorcr
  \email{sadutta@mpi-inf.mpg.de}}
\affil[4]{University of Regensburg, Regensburg, Germany\authorcr
  \email{gorav.jindal@ur.de}}

\shorttitle{Shift Equivalence Testing for Sparse Polynomials}
\shortauthor{Bl\"aser et al.}
\shortdate{~}
\novelty{}
\date{}

\begin{document}

\abstract{
We study the shift equivalence testing problem for high-degree sparse multivariate polynomials. Given explicit sparse representations of two polynomials $f,g\in\Qset[x_1,\ldots,x_n]$, with exponents encoded in binary, sparsities $s$ and $t$, and total degree at most $d$, the task is to decide whether there is a vector $\ba\in\Qset^n$ such that $f(\bx+\ba)\equiv g(\bx)$.
Furthermore, one needs to output such a vector if it exists.
Dvir, Oliveira, and Shpilka (ICALP 2014) studied shift equivalence for low-degree arithmetic circuits. Their framework does not apply directly to sparse polynomials, which are not closed under shifts, and they left open whether shift equivalence of sparse polynomials admits an efficient deterministic algorithm.
We answer this question positively over $\Qset$. If the input coefficients have bit-size at most $\tau$, we show that a rational shift of polynomial bit-size exists whenever the polynomials are shift equivalent. Moreover, our algorithm outputs such a shift and has a running time of $\poly(n,s,t,\tau,\log d)$ in the bit model. The arithmetic and structural parts of our result also hold over any field of characteristic zero or characteristic larger than $d$.
}

\maketitle
\hypersetup{
  pdfauthor={Markus Bl\"aser, Pranjal Dutta, Sagnik Dutta, Gorav Jindal}
}

\section{Introduction}

Shift Equivalence Testing (SET) asks the following seemingly simple question: Given two
polynomials $f,g\in\Qset[x_1,\ldots,x_n]$, decide whether there is a vector
$\ba=(a_1,\ldots,a_n)\in\Qset^n$ such that
\[
        f(\bx+\ba)\equiv g(\bx),
\]
and output such a vector if it exists. 
The symbol $\equiv$ emphasizes that the
equality is an identity of formal polynomials, not merely an equality of functions. Over $\Qset$ these two notions coincide. The same is true, and the Taylor expansions used below remain valid, over any field of characteristic zero or characteristic larger than the degree under consideration. Note that shift equivalence is indeed an equivalence relation. It is reflexive, since $f$ can be transformed into itself by the zero shift. When $f(\bx + \ba) = g(\bx)$, then also $f(\bx) = g(\bx- \ba)$, so it is symmetric. And if shifting by $\ba$ transforms $f$ into $g$ and shifting by $\bb$ transforms $g$ into $h$, then $\ba + \bb$ transforms $f$ into $h$, giving transitivity.

The algorithmic problem of Shift Equivalence Testing first appeared in computer algebra and symbolic computation.
Grigoriev \cite{grigorievtcs1997} designed the first polynomial-time algorithms for this problem. He assumes that the polynomials are given in dense representation, as a list of \emph{all} coefficients including those that are zero. For an $n$-variate polynomial of total degree at most $d$, this list contains $\binom{n+d}{d}$ coefficients. His deterministic characteristic-zero algorithm takes $n^{O(d)}$ time, so zero-padding a sparse input does not give a running time polynomial in the sparse parameters. Grigoriev also gives a randomized algorithm for large enough finite characteristic and a quantum algorithm for characteristic two.


For polynomials with few nonzero coefficients, a running time polynomial in the dense input size may be far larger than polynomial in the sparse input length.
Dvir, Oliveira, and Shpilka \cite{dviroliveirashpilkai2014icalp} studied SET for polynomials represented compactly by arithmetic circuits. For instance, the determinant of an $n\times n$ matrix has $n!$ nonzero monomials and an arithmetic circuit of polynomial size. Dvir et al.\ give a randomized algorithm for SET with running time polynomial in the circuit size and the total degree $d$. Their algorithm also works with black-box access to $f$ and $g$ and bounds on their degree and circuit size. It proceeds by a reduction to polynomial identity testing.

The polynomial identity testing problem (PIT) is a famous problem in complexity theory. Given a polynomial $f \in \Qset[\bx]$ as an arithmetic circuit $C$, is $f$ identically zero? There is a simple randomized polynomial-time algorithm for the problem based on the so-called Schwartz--Zippel lemma \cite{schwartz1980, zippel1979eurosam}, which essentially inserts a random point from a large enough set and checks whether it evaluates to zero. Derandomizing PIT is a big open problem in complexity theory. In particular, it implies non-uniform circuit lower bounds \cite{kabanetsimpagliazzo2004}, another major open problem in the field. Note that the algorithm by Dvir et al.\ shows that SET reduces to PIT under Turing reductions. But also the converse is true. The polynomial $f$ is identically zero iff it is shift-equivalent to the zero polynomial. Therefore, a full derandomization of SET is out of reach of current methods.

Since finding a deterministic algorithm for general PIT is a hard problem, researchers have developed deterministic identity tests for restricted circuit classes; see, e.g., \cite{DuttaGhosh2024PITSurvey}. For read-once oblivious algebraic branching programs (ROABPs), white-box PIT is deterministic polynomial time \cite{razshpilka2005}, while deterministic black-box PIT is known in quasipolynomial time \cite{forbesshpilka2013,forbessaptharishishpilka2014}. The framework of Dvir, Oliveira, and Shpilka \cite{dviroliveirashpilkai2014icalp} therefore gives the corresponding white-box and black-box SET algorithms for ROABPs.

Their deterministic reduction requires closure properties. Sparse polynomials are not closed under shifts, so it does not apply directly and DOS'14 states no deterministic sparse-class algorithm. Indirectly, however, representing both explicit sparse inputs as ROABPs in one common variable order and using white-box ROABP PIT gives deterministic SET with $\poly(n,s,t,d)$ arithmetic complexity. This is not polynomial in the binary sparse input size when $d$ is exponential.

Our contribution is a deterministic algorithm polynomial in the length of explicit sparse input over $\Qset$. We now specify this input model. Each polynomial is given as an explicit list of nonzero terms
\[
c\,x_1^{\alpha_1}\cdots x_n^{\alpha_n},
\qquad c\in\Qset,
\]
where the coefficient is given explicitly in binary representation. 
Furthermore, the exponents are also encoded in binary. We call this an \emph{explicit sparse representation}.
Thus, the natural input parameters are the number of variables $n$, the sparsities $s=\sp(f)$ and $t=\sp(g)$, the height, equivalently bit-size, of the coefficients of $f,g$, and $\log d$, where $d$ bounds the total degrees of $f$ and $g$. (Sometimes, the maximum bit-size of the coefficients is called logarithmic height, but we use the term height for brevity.)
While low degree is conventional in algebraic complexity, high degree is natural for sparse or lacunary polynomials: an exponent bounded by $d$ uses only $O(\log d)$ bits.
Classical sparse interpolation likewise targets this regime: Ben-Or--Tiwari has algebraic-RAM complexity polynomial in $\log d$ \cite{benortiwari1988stoc}, and Garg--Schost obtains such dependence under straight-line-program access \cite{gargschost2009tcs}.
The indirect deterministic ROABP algorithm described above has polynomial dependence on $d$, so it does not give polynomial time in the binary sparse input length.\footnote{For this reason, some sparse representations encode exponents in unary. \cite{KaltofenKoiran2005} uses the terms lacunary or supersparse when exponents are encoded in binary; here, sparse representation always uses binary exponents.} Our algorithm reduces the degree dependence to $\poly(\log d)$ for explicit sparse inputs and also provides polynomial bit complexity over $\Qset$.

More formally, let {\textbf{Sparse SET}} denote the following problem: given two polynomials $f,g\in\Qset[x_1,\ldots,x_n]$ in sparse representation, decide whether there exists a shift $\ba\in\Qset^n$ such that
\[
f(\bx+\ba)\equiv g(\bx).
\]
%
%
%
As mentioned above, PIT for sparse polynomials is trivial in the sparse representation: one simply compares the monomial list with the zero list. The issue is now that the sparsity is not preserved under shifts. A single monomial can become very large: if we define $\supp(\ba)\coloneqq \{i:a_i\ne 0\}$, then
\[
(\bx+\ba)^{\bgamma}
=\prod_{i=1}^n (x_i+a_i)^{\gamma_i}
\]
has exactly
\[
\prod_{i\in\supp(\ba)}(\gamma_i+1)
\]
nonzero monomials. This number can be exponential in the sparse input length, already for a single monomial and even when the exponents are given in unary. Thus, even if $f$ is sparse, the shifted polynomial $f(\bx+\ba)$ may be dense. Dvir, Oliveira, and Shpilka~\cite[Section 6]{dviroliveirashpilkai2014icalp} pointed out precisely this obstruction: their methods ``almost'' give a quasipolynomial-time algorithm for shift equivalence of sparse polynomials, but the direct argument breaks down because sparse polynomials are not closed under shifts. For explicit sparse inputs, the indirect ROABP approach described above gives a deterministic algorithm using $\poly(n,s,t,d)$ arithmetic operations, which is not polynomial in the sparse input length when $d$ is exponential. This leads to the central question addressed in this paper.

\begin{problem}
Can Sparse SET over $\Qset$ be solved using $\operatorname{poly}(n,s,t,\log d)$ arithmetic operations and bit complexity $\operatorname{poly}(n,s,t,\log d,\tau)$?
\end{problem}

In this paper, we answer the question affirmatively. The first key ingredient is a new structural lemma showing that, in yes instances, the witnessing shift is very special, which forces the relevant intermediate polynomials to remain sparse. This addresses the problem that sparse polynomials are not closed under shifts.

The second problem stems from the fact that the algorithm in~\cite{dviroliveirashpilkai2014icalp} has running time polynomial in $d$, so there is no easy way of adapting this algorithm to a running time of $\poly(\log d)$. The algorithm of~\cite{dviroliveirashpilkai2014icalp} iteratively solves a system of equations yielding certain intermediate polynomials. We modify this approach in a nontrivial manner by showing that it is enough to store and compare all intermediate polynomials in the sparse representation without expanding arbitrary shifts.

There are also further obstacles. Even in yes instances over $\Qset$, one must first show that there exists a witnessing shift $\ba$ in $\Qset^n$, rather than only over an algebraic extension\footnote{Consider $f(\bx+\bz)-g(\bx) = \sum_{\alpha} F_{\alpha}(\bz)\bx^{\alpha}$. Thus, one needs to find a common solution of $F_{\alpha}(\bz)$, and the weak Nullstellensatz gives a common solution~$\ba \in \overline{\Qset}^n$; for more details, see~\cref{small-shift}}. One must then prove that such a shift has polynomially bounded bit complexity in the input size, and for the search version, compute and output it within the same bounds.


\subsection{Our results}

Our main theorem is a deterministic algorithm for shift equivalence of explicitly given sparse multivariate polynomials over the rationals, with running time polynomial in the sparse input parameters. We also prove polynomial bit complexity and output a shift of polynomial bit-size.

For an extension field $\mathbb K\supseteq \Qset$, define
\[
        S_{f,g}(\mathbb K)
        \coloneqq
        \{\ba\in \mathbb K^n : f(\bx+\ba)\equiv g(\bx)\}.
\]
In particular, $S_{f,g}(\overline{\Qset})$ denotes the set of shifts over the
algebraic closure. The rationality result below shows that if
$S_{f,g}(\overline{\Qset})$ is nonempty, then it contains a rational point.
Thus, for the rational decision problem, we may write
\[
        S_{f,g}:=S_{f,g}(\Qset)
\]
without changing the decision problem. 

We state the main results over $\Qset$ to keep the notation simple. The
structural statements and arithmetic algorithms are purely algebraic and remain
valid over any field $\mathbb K$ of characteristic zero or characteristic larger
than $d$, after replacing $\Qset$ by $\mathbb K$ and omitting the height and
bit-complexity assertions.

\begin{theorem}[Main theorem]
\label{thm:intro-main}
Let $f,g\in\Qset[x_1,\ldots,x_n]$ be given in explicit sparse representation, with
sparsities $s$ and $t$ and degree at most $d$. There is a deterministic
algorithm using $\poly(n,s,t,\log d)$ arithmetic operations that decides whether
$S_{f,g}$ is nonempty. If $S_{f,g}\neq\emptyset$, the algorithm outputs a shift
$\ba\in S_{f,g}$.

If the input coefficients have height at most $\tau$, then one can
choose the output shift to have height
$\poly(n,s,t,\tau,\log d)$, and the algorithm can be implemented using
$\poly(n,s,t,\tau,\log d)$ bit operations.
\end{theorem}

The arithmetic algorithm is proved in \cref{thm:sparse-shift-correctness}. The
bit-complexity analysis over $\Qset$ is given in
\cref{prop:intermediate-height}, and the polynomial-height output guarantee
uses the small rational witness theorem, \cref{small-shift}. This exponentially improves on the indirect ROABP consequence of the DOS'14 framework described above, whose running time has polynomial dependence on $d$, rather than on $\log d$. In
addition, that framework does not by itself address the existence of a rational shift of polynomial bit-size. Thus both the $\poly(n,s,t,\log d)$ arithmetic
bound and the bit-complexity guarantee require additional ideas.

\paragraph{Black-box access.}
Over $\Qset$, given separate evaluation oracles for $f,g$ promised to be sparse,
Ben-Or--Tiwari interpolation \cite{benortiwari1988stoc} recovers their explicit
representations using $\poly(n,s,t,\log d)$ queries and algebraic-RAM operations.
This model permits unit-cost integer remainder operations in addition to
arithmetic, and oracle calls are counted at unit cost. Applying
\cref{thm:intro-main} to the reconstructed representations gives the same bound
in this model. This does not imply the same bit bound: evaluating $x^d$ at $2$ already returns a
$(d+1)$-bit integer. For integer coefficients of magnitude at most $H$, the
modular oracle of Bläser--Jindal \cite{blaserjindal2014issac} yields deterministic
$\poly(n,s,t,\log d,\log H)$ bit-complexity interpolation. Division-free
integer circuits simulate this oracle by evaluating each gate modulo the queried
modulus; the cost of this simulation also depends on the circuit description size.

The first key technical ingredient is the following sparsity bound. It rules out the possibility that, in a yes instance, the sparse output is obtained only by
cancellations among large shifted monomial expansions.

\begin{theorem}[Structural sparsity bound]
\label{thm:intro-structural}
Let $f,g\in\Qset[x_1,\ldots,x_n]$ have sparsities $s$ and $t$, respectively,
and degree at most $d$. If $\ba\in S_{f,g}$, then for every monomial
$\bx^{\bgamma}$ appearing in $f$ or
in $g$,
\[
        \prod_{i\in\supp(\ba)}(\gamma_i+1)\le st .
\]
\end{theorem}

In fact, the theorem we prove is much stronger than the statement above. It applies already when $g(\bx)$ and $f(\bx+\ba)$ agree only in degrees at least some
threshold $r$. In that case, for every monomial of $f$ of degree at least $r$, the part of its shifted expansion that can contribute to degrees at least $r$
has sparsity at most $st$. This truncated form is what is used during the
degree-by-degree descent in the proof of \cref{thm:intro-main}; see
\cref{lem:truncated-structural}. The stronger version is needed for our analysis to work.

\begin{remark}
The inequality in \cref{thm:intro-structural} does not bound the total degree
of $f$ or $g$. It only bounds the exponents in the coordinates in which the
shift is nonzero. This is the relevant restriction for sparse representation:
variables outside $\supp(\ba)$ do not create binomial expansions, while a
monomial $\bx^{\bgamma}$ expands into
\[
        \prod_{i\in\supp(\ba)}(\gamma_i+1)
\]
terms in the shifted coordinates. Thus, the theorem does not say anything about how large the exponents in variables outside the support of the shift could be; it rules out a large binomial expansion in the shifted coordinates.
\end{remark}

The second ingredient, needed for bit complexity over $\Qset$, is that a small
rational witness always exists. This is not automatic from the existence of an
arbitrary shift: a priori, the solution set could contain only shifts over an
algebraic extension, or contain rational shifts only with very large
coordinates. We prove that neither obstruction occurs for sparse SET.

\begin{theorem}[Small rational witness]
\label{thm:intro-small-shift}
\label{small-shift}
Let $f,g\in\Qset[x_1,\ldots,x_n]$ have sparsities $s$ and $t$, degree at most
$d$, and coefficient height at most $\tau$. If
$S_{f,g}(\overline{\Qset})\neq\emptyset$, then $S_{f,g}(\Qset)$ contains a
shift whose coordinates have height
\[
        \poly(n,s,t,\tau,\log d).
\]
Moreover, the polynomial bound is explicit.
\end{theorem}

The proof of \cref{thm:intro-small-shift} is deferred to
\cref{app:missing-proofs}. It uses the structural sparsity bound again. Once a shift
exists, the structural theorem restricts the derivatives of
$f$ in the shifted coordinates to a polynomial-dimensional space. Inside this
space, the translation from $f$ to $g$ can be represented by a rational
differential operator of polynomial height. Taking a finite logarithm of this
operator and solving one linear system gives a rational shift of polynomial
bit-size.

As a further corollary of the structural bound, we obtain a sparse-time
white-box PIT algorithm for shifted-sparse expressions. Here the shifts are
given explicitly.

\begin{corollary}[White-box PIT for shifted-sparse polynomials]
\label{cor:intro-verify-shift}
Let $f,g\in\Qset[x_1,\ldots,x_n]$ be sparse polynomials of sparsities $s$ and
$t$ and degree at most $d$, and let $\ba,\bb\in\Qset^n$ be given. There is a
deterministic algorithm using
$\poly(n,s,t,\log d)$ arithmetic operations that decides whether
\[
        f(\bx+\ba)\equiv g(\bx+\bb).
\]
Further, if the input coefficients as well as the vectors $\ba,\bb$ have height at most $\tau$, then the algorithm can be implemented using $\poly(n,s,t,\tau,\log d)$ bit operations.
\end{corollary}

The proof is given after \cref{lem:truncated-structural} in
\cref{sec:struc-lem}.


\begin{remark}
Without the structural theorem, the immediate deterministic approach gives only a bound that depends polynomially on $d$. Indeed,
$f(\bx+\ba)-g(\bx+\bb)$ can be represented as a sum of products of univariate
polynomials, and hence as a read-once oblivious algebraic branching program.
Applying the white-box PIT algorithm of Raz and Shpilka
\cite{razshpilka2005} gives a deterministic verification algorithm using
$\poly(n,s,t,d)$ arithmetic operations. 
The corollary above improves this to $\poly(n,s,t,\log d)$ by using the structural bound to keep the computation in sparse representation. 
\end{remark}

\subsection{Proof Overview}

We describe the main ideas behind the algorithm and the analysis. The algorithm is based on the homogeneous descent of~\cite{dviroliveirashpilkai2014icalp}. Their framework gives efficient algorithms for low-degree arithmetic circuits, and in particular, its running time is polynomial in the degree bound. For sparse polynomials, however, the degree can be exponentially larger than its binary encoding, so polynomial dependence on $d$ is not sufficient.

Two additional issues have to be handled in the sparse setting. First, a sparse polynomial need not remain sparse after a shift: even $(x+a)^d$ has $d+1$
monomials, which is exponential in $\log d$. In several variables the blow-up can be $(d+1)^n$. Thus we cannot simply expand shifted monomials during the
descent. Second, over $\Qset$, an arithmetic algorithm is {\em not enough}: one must prove that a polynomial-height rational shift exists and that the algorithm can
output such a small witness.
\paragraph{The \cite{dviroliveirashpilkai2014icalp} descent.}
Write $f_i:=H^i(f)$ and $g_i:=H^i(g)$ for the homogeneous components of degree
$i$. For a vector $\bv\in\Qset^n$, let
\[
        D_{\bv}:=\sum_{\ell=1}^n v_\ell\frac{\partial}{\partial x_\ell},
        \qquad
        f_i^{(r)}(\bv,\bx):=D_{\bv}^r f_i(\bx).
\]
So, $D_{\bv}$ denotes the directional derivative in direction $\bv$. Moreover, $f_i^{(1)}(\bv,\bx)$ is the first-order contribution of the degree-$i$
part of $f$ under a shift in direction $\bv$, while
$f_i^{(r)}(\bv,\bx)$ is the corresponding $r$th-order contribution. Suppose first that $g(\bx)=f(\bx+\bc)$ for some shift $\bc$. The homogeneous
Taylor expansion is {\em triangular} with respect to degree: the degree-$i$ part of
$g$ receives contributions only from the degree-$i$ and higher homogeneous parts
of $f$. More explicitly,
\[
        g_i
        =
        f_i
        +
        f_{i+1}^{(1)}(\bc,\bx)
        +
        \sum_{j=i+2}^{d}
        \frac{1}{(j-i)!} f_j^{(j-i)}(\bc,\bx).
\]
The term $f_i$ is known, and the term
$f_{i+1}^{(1)}(\bc,\bx)$ is linear in the unknown shift $\bc$. The remaining
sum is a higher-order tail coming from degrees strictly larger than $i+1$. This suggests a {\em top-down descent}. At degree $d-1$, there is no tail, so
$$
        g_{d-1}=f_{d-1}+f_d^{(1)}(\bc,\bx)
$$
is just a linear system in the coordinates of $\bc$. At degree $d-2$, the
equation contains the extra term $\frac{1}{2}f_d^{(2)}(\bc,\bx)$, which looks {\em quadratic} in $\bc$. The key observation of~\cite{dviroliveirashpilkai2014icalp} is that this term is already determined once the first derivative of $f_d$ has been matched.
Indeed, if two directions $\ba$ and $\bc$ satisfy
$$
        f_d^{(1)}(\ba,\bx)=f_d^{(1)}(\bc,\bx),
$$
then $D_{\ba-\bc}f_d=0$. Since constant-coefficient differential operators commute, this implies that all higher directional derivatives of $f_d$ in the directions $\ba$ and $\bc$ also agree. Thus the degree-$(d-2)$ equation again
has only one new unknown contribution, namely $f_{d-1}^{(1)}(\bc,\bx)$,
which is {\em linear} in $\bc$. More generally, at every stage of the descent, all higher-order terms coming from previously processed homogeneous components are
already fixed, and the only new condition to impose is a linear equation for the first derivative of the next homogeneous component.

\paragraph{Modification in the sparse setting.}~We apply this descent only at degrees where a new first-order term can appear.
Let
\[
        d=\delta_1>\delta_2>\cdots>\delta_m>0
\]
be the nonzero positive homogeneous degrees of $f$. Since $f$ has sparsity $s$,
we have $m\le s$. At stage $q$, the algorithm has already processed
$\delta_1,\ldots,\delta_{q-1}$ and has a candidate shift $\ba_{q-1}$ satisfying
\[
        f_{\delta_p}^{(1)}(\ba_{q-1},\bx)
        =
        f_{\delta_p}^{(1)}(\bc,\bx)
        \qquad\text{for all }p<q.
\]
By the derivative observation above, the higher-order contributions of these
already processed components are the same for $\ba_{q-1}$ and for the true
shift $\bc$. Therefore the {\em tail}
\[
        P_q(\bx)
        :=
        H^{\delta_q-1}
        \left(
            H^{\ge \delta_q+1}(f)(\bx+\ba_{q-1})
        \right)
\]
is equal to the tail one would obtain from $\bc$:
\[
        P_q(\bx)
        =
        H^{\delta_q-1}
        \left(
            H^{\ge \delta_q+1}(f)(\bx+\bc)
        \right).
\]
After subtracting this tail, the next condition is again linear:
\[
        g_{\delta_q-1}
        =
        f_{\delta_q-1}
        +
        f_{\delta_q}^{(1)}(\bb,\bx)
        +
        P_q(\bx).
\]
The algorithm adds this linear constraint to the previous ones and chooses a
solution $\ba_q$. Thus the nonlinear shift-equivalence problem is reduced, one
homogeneous degree at a time, to simple linear algebra.

\paragraph{Why this is not yet an efficient algorithm.}
The discussion above explains why the shift problem can be reduced to a sequence of linear systems. It does not yet explain why these systems can be constructed and eventually solved from the sparse input in polynomial time. The difficulty lies in the computation of the tail $P_q$. A single monomial $\bx^{\bgamma}$, after the shift $\bx \mapsto \bx + \ba$, produces $\prod_{i\in \operatorname{supp}(\ba)}(\gamma_i+1)$ many monomials, which may be exponential in the sparse input length. Thus, explicitly forming shifted monomials is too expensive.

One could keep shifted monomials in a circuit representation, but then the algorithm has left the sparse representation. The linear systems in the descent
would have entries represented by circuits (of high degree), and constructing or comparing those entries would require PIT for a larger circuit class. The general framework of
\cite{dviroliveirashpilkai2014icalp} then gives dependence polynomial in the degree $d$, or in the relevant circuit/PIT parameters, rather than a polynomial in $\log d$.

There is a second issue over $\Qset$. A naive bit-complexity analysis would
bound the tail $P_q$ in terms of the height of the current candidate $\ba_{q-1}$, since $P_q$ is computed using $\ba_{q-1}$. The right-hand side of the next linear system would then have height depending on $\Ht(\ba_{q-1})$, and
standard bounds for basic rational solutions (defined in \cref{app:missing-proofs}, \cref{def:basic-rational-solution}) would give a recursive estimate for $\Ht(\ba_q)$ in terms of $\Ht(\ba_{q-1})$. Iterating such an estimate over up to $s$ stages need not yield the desired polynomial bound in the sparse input parameters. Thus, the tails must be bounded globally, using a fixed small witnessing shift, rather than recursively through the candidates produced by the algorithm.

\paragraph{The structural lemma and proof sketch of~\cref{thm:intro-structural}.}
The main new structural input is that the bad shifted expansions cannot occur
along a yes instance. The truncated form of the lemma says the following.
Suppose
\[
        H^{\ge r}(g)=H^{\ge r}(f(\bx+\ba)).
\]
Then for every monomial $\bx^{\bgamma}$ of $H^{\ge r}(f)$,
\[
        \sp\left(H^{\ge r}\big((\bx+\ba)^{\bgamma}\big)\right)\le st.
\]
Thus, once the shifted polynomial agrees with $g$ in degrees at least $r$, every shifted monomial that contributes to those degrees is itself sparse {\em after truncation}. In particular, large shifted expansions cannot be hidden by
cancellations in the final sparse polynomial.

The proof isolates one shifted monomial from the identity
$g(\bx)=f(\bx+\ba)$. Let $f(\bx)=\sum_{j=1}^s c_j \bx^{\balpha_j}$, and relabel
the target monomial so that $\bgamma=\balpha_1$. We use Euler-type differential operators, defined as follows:
$$
        E_i=(x_i+a_i)\partial_{x_i}.
$$
These operators are diagonal on the shifted monomial basis: $E_i(\bx+\ba)^\be=e_i(\bx+\ba)^\be$. Hence, for a multi-index $\bv$, the following holds:~$\bE^\bv(\bx+\ba)^{\balpha_j}\,=\,\balpha_j^\bv(\bx+\ba)^{\balpha_j}$. Thus the operators $\bE^\bv$ act on shifted monomials exactly as evaluation functions on the exponent vectors $\balpha_j$.

Using interpolation on the finite set of exponent vectors $\{\balpha_1,\ldots,\balpha_s\}$, we choose a set $T$ of $s$ multi-indices with the following {\em cone-closed} property: if $\bv\in T$ and $\bu\preceq\bv$ coordinatewise, then $\bu\in T$. The concept of a cone-closed basis was introduced in~\cite{forbes2018towards}. We also choose coefficients $r_\bv$ such that
$$
\sum_{\bv\in T} r_\bv \balpha_j^\bv
=
\begin{cases}
1 & j=1,\\
0 & j\ne 1.
\end{cases}
$$
Now define the isolating operator
\[
        \mathcal I_{\bgamma}:=\sum_{\bv\in T} r_\bv \bE^\bv .
\]
By construction, $\mathcal I_{\bgamma}$ {\em kills} all shifted monomials except the one with exponent
$\bgamma$. Therefore $\mathcal I_{\bgamma}(g)=c_1(\bx+\ba)^{\bgamma}$.
The same operator $\mathcal I_{\bgamma}$ can also be viewed as acting on the ordinary sparse expansion of $g$. This is exactly where {\em sparsity enters}. Since $E_i=x_i\partial_{x_i}+a_i\partial_{x_i}$, each application of $E_i$ to an ordinary monomial either preserves its exponent or decreases it by one in coordinate $i$. More generally, applying $\bE^\bv$ to a
monomial $\bx^{\bbeta}$ can only produce monomials of the form $\bx^{\bbeta-\bu}$ with $\bu\preceq\bv$.

Because $T$ is cone-closed, every such possible exponent drop $\bu$ is again contained in $T$. Hence, if $g(\bx)=\sum_{k=1}^t d_k\bx^{\bbeta_k}$, then the support of $\mathcal I_{\bgamma}(g)$ is contained in $\{\bbeta_k-\bu : k\in[t],\ \bu\in T\}$. This set has size at most $t|T|=st$. Since $\mathcal I_{\bgamma}(g)$ is a nonzero scalar multiple of $(\bx+\ba)^{\bgamma}$, the shifted monomial $(\bx+\ba)^{\bgamma}$ itself has support at most $st$.

The truncated version needed by the algorithm follows from the same argument. Indeed, the operators $E_i$ never increase total degree, so the isolation
argument can be applied directly to the high-degree identity
\[
        H^{\ge r}(g)=H^{\ge r}(f(\bx+\ba)).
\]
This gives the same support bound for the truncated shifted monomials $H^{\ge r}((\bx+\ba)^{\bgamma})$.

This is the point where our proof differs from the general framework of
\cite{dviroliveirashpilkai2014icalp}. Sparse polynomials are not closed under arbitrary shifts. The structural lemma saves us here: although an arbitrary shift of a sparse polynomial may be dense, the shifted
pieces that arise along a yes instance of the descent are forced to remain sparse. For details, see~\cref{lem:truncated-structural}.

\paragraph{Computing the tails without full expansion.}
At the beginning of stage $q$, the maintained invariant gives
\[
        H^{\ge \delta_q}(f(\bx+\ba_{q-1}))
        =
        H^{\ge \delta_q}(g).
\]
Therefore the structural lemma applies with $r=\delta_q$ and
$\ba=\ba_{q-1}$. Thus, in particular, for every monomial
$c_{\bgamma}\bx^{\bgamma}$ in $H^{\ge\delta_q}(f)$, the high-degree truncation
\[
        H^{\ge \delta_q}
        \big(c_{\bgamma}(\bx+\ba_{q-1})^{\bgamma}\big)
\]
has at most $st$ monomials in a yes instance.

The obvious way to use this statement would be to expand
$c_{\bgamma}(\bx+\ba_{q-1})^{\bgamma}$ and then keep only the terms of degree at
least $\delta_q$. Unfortunately, this is {\em not efficient}: the shifted monomial may be dense as a
whole, even though its high-degree truncation is sparse. Nor can one simply
invoke sparse interpolation as a black box, since the polynomial naturally
available for evaluation is the entire shifted monomial, whose lower-degree
part may be dense. The sparse object guaranteed by the structural lemma is {\textbf{only}} the high-degree truncation.

Instead, the algorithm enumerates only the exact degree-$\delta_q$ layer, which
is part of the structurally bounded high-degree truncation. Using a truncated
dynamic-programming enumeration, the algorithm either lists all drops that
produce degree $\delta_q$, if there are at most $st$, or detects that there are
more than $st$ and rejects. Thus it never enumerates the potentially dense
lower-degree part; for details see~\cref{lem:compute-Pik}.

The desired tail $P_q$, however, lies one degree lower, in degree
$\delta_q-1$. The algorithm does not enumerate this lower layer directly. Let
$j=|\bgamma|$ and let
\[
        C_\gamma(\bx)
        :=
        H^{\delta_q}
        \big(c_{\bgamma}(\bx+\ba_{q-1})^{\bgamma}\big).
\]
Then the degree-$(\delta_q-1)$ contribution of
$c_{\bgamma}(\bx+\ba_{q-1})^{\bgamma}$ is obtained from the controlled degree
$\delta_q$ part by
\[
        H^{\delta_q-1}
        \big(c_{\bgamma}(\bx+\ba_{q-1})^{\bgamma}\big)
        =
        \frac{1}{j-\delta_q+1}\,
        D_{\ba_{q-1}} C_\gamma(\bx).
\]
Thus the algorithm enumerates only this structurally bounded layer and then
obtains the one lower layer needed for the descent by a single directional
derivative. Summing these contributions over the monomials of $H^{>\delta_q}(f)$
and collecting like terms gives $P_q$ in sparse representation.

\paragraph{Bit complexity assuming a small rational shift.}
We first explain the bit-complexity analysis under the following assumption: there is a true shift $\bc\in S_{f,g}$ whose height is polynomially bounded in the input parameters. Later, we will explain why such a shift always exists over $\Qset$ whenever $S_{f,g}$ is nonempty.

Under this assumption, the key point is that the tails are not recursively bounded. As mentioned before, if one does it recursively, then iterating over the stages of the descent would give an exponential bound on the output height. The way around this is to prove that the tails are controlled by a single global witness. Although $P_q$ is computed syntactically using $\ba_{q-1}$, the descent invariant gives
\[
        P_q(\bx)
        =
        H^{\delta_q-1}
        \left(
            H^{\ge \delta_q+1}(f)(\bx+\bc)
        \right)
\]
for one fixed true shift $\bc$. Hence, after collecting terms, the height of
$P_q$ is controlled in terms of the input and the height of this same $\bc$,
not in terms of a sequence of previously computed candidates. Once the tails
have polynomial height, the linear systems have polynomial-height coefficients,
and choosing basic rational solutions gives polynomial-height candidates; the
definition and the height bound are in \cref{app:missing-proofs},
\cref{def:basic-rational-solution,lem:basic-rational-solution-height-bit}. The
algorithm also imposes global support checks and coefficient-height checks on
the computed tails throughout the descent. These checks are harmless in a yes
instance, by the structural lemma and the small-witness bound. On a no
instance, the algorithm may reject because one of these checks fails or because
one of the linear systems is infeasible.
If neither happens, the algorithm obtains a candidate shift $\ba$. It then
performs a final verification of the identity $f(\bx+\ba)\equiv g(\bx)$.

This final check can again be done very efficiently: by the structural lemma, if the identity holds, then every shifted monomial of $f(\bx+\ba)$ has support at most $st$ after truncation with $r=0$, and hence the whole expansion has support at most $s^2t$. The algorithm expands these bounded pieces, collects like terms, and compares the resulting sparse polynomial with $g$. Thus, the algorithm never
accepts a false candidate. For the detailed analysis, we refer to~\cref{prop:intermediate-height}.

\paragraph{Existence of a small rational shift.}~It remains to justify that a small rational shift exists whenever any shift
exists. First, if a shift exists over an algebraic extension, then a rational
shift exists over $\Qset$: the complex shift set $S_{f,g}(\Cset)=\{\ba\in\Cset^n:f(\bx+\ba)=g(\bx)\}$ is an affine subspace, and averaging an algebraic point over its Galois conjugates gives a rational point in the same affine subspace.

The height bound uses the structural lemma once more. Let $\bc$ be any rational shift and let $S=\operatorname{supp}(\bc)$. Then every monomial
$\bx^{\bgamma}$ of $f$ satisfies $\prod_{i\in S}(\gamma_i+1)\le st$. Consequently, the derivatives of $f$ in the variables of $S$ span only a
polynomial-dimensional space. Inside this space, the translation $f(\bx)\mapsto f(\bx+\bc)=g(\bx)$ can be represented by a rational differential
operator of polynomial height; see the proof of \cref{small-shift} in
\cref{app:missing-proofs}, in particular \cref{differential-operator}. This operator is {\em unipotent} on the relevant derivative space, so its logarithm is a finite sum. The logarithm recovers the first-order translation direction on $f$, namely a polynomial of the form $D_{\bc}f$. 

Solving the linear system $D_{\bb}f = D_{\bc}f$
and taking a basic rational solution using \cref{lem:basic-rational-solution-height-bit} in \cref{app:missing-proofs} gives a vector $\bb$ of polynomial height.
By the derivative characterization of shifts from
\cite{dviroliveirashpilkai2014icalp}, this $\bb$ is again a shift from $f$ to $g$. Thus, whenever Sparse SET is a yes instance over $\Qset$, there is a witness of polynomial bit-size, and the algorithm outputs one. For details, refer to~\cref{small-shift}; its proof is deferred to \cref{app:missing-proofs}.

\subsection{Related work} \label{sec:related-work}

\paragraph{Algorithms related to sparse polynomials.}
Several basic algorithmic problems have been studied for sparse polynomials, including identity testing,
interpolation, factorization, and changes of origin. Sparse identity testing and interpolation go back to the works of Zippel and Ben-Or--Tiwari, and have been
followed by several deterministic and randomized algorithms for sparse polynomials and straight-line programs
\cite{zippel1979eurosam,benortiwari1988stoc,klivansspielman2001stoc, blaserdhardtliptonvishnoi2009ipl,gargschost2009tcs,
arnoldgiesbrechtroche2013casc}. Sparse factorization has also been extensively studied, beginning with algorithms for bounded-degree factors of sparse
polynomials and followed by more recent structural and algorithmic results for factors of sparse polynomials
\cite{lenstra1999lacunary,avendanokricksombra2006issac,
bhargavasarafvolkovich2020jacm,chuyoonshpilka2026sparse,
bhattacharjeekotharyraisaraf2026sparse}. For a survey-style introduction to the algorithmic landscape for sparse polynomials, including arithmetic, interpolation, and factorization, see Roche~\cite{roche2018sparse}.

\paragraph{Sparse shifts and sparsification.}
The algorithmic study of shifts of sparse polynomials goes back to Grigoriev and Karpinski, Lakshman and Saunders, and Grigoriev and Lakshman
\cite{grigorievkarpinski1993aaecc,lakshmans94issac,
grigorievlakshman1995issac}. Their work considered questions such as identity
testing, interpolation, and finding changes of origin under which a polynomial
has a smaller sparsity. Subsequent work studied algorithms for computing sparsest shifted representations and for interpolation in shifted-sparse models
\cite{giesbrechtkaltofenlee2003jsc,giesbrechtroche2010cc}.

\paragraph{Target-free sparse shifts and sparsification.}
A related but different problem is the {\em target-free sparse-shift} problem: given a
polynomial $f$, decide whether there exists a shift $\ba$ such that
$\sp(f(\bx+\ba))$ is small, or smaller than $\sp(f)$. Recent work shows that this
sparsification problem is computationally hard
\cite{chillaragrichenershpilka2023stacs,blaserduttajindal2026soda}. This should be contrasted with Sparse SET: here both endpoints $f$ and $g$ are part of the
input, and the question is whether $g(\bx)=f(\bx+\ba)$ for some shift $\ba$. Thus, our problem is target-given rather than target-free.

\paragraph{Polynomial equivalence and orbit testing.}
There is also a broader polynomial-equivalence literature in which one allows
general linear or affine changes of variables. In this setting one asks whether
there are $A\in \mathsf{GL}_n(\Qset)$ and $\bb\in\Qset^n$ such that $f(A\bx+\bb) \equiv g(\bx)$, or whether the orbit of a given polynomial intersects a prescribed structured
class. Polynomial equivalence is often viewed as an algebraic analogue of graph
isomorphism: already, the equivalence of homogeneous cubic forms captures graph isomorphism
\cite{agrawalsaxena2006stacs,kayal2011soda}. Kayal studied efficient algorithms for special cases of polynomial equivalence and later affine projections of
polynomials \cite{kayal2011soda,kayal2012stoc}. 

More recently, there has been work on equivalence testing for structured polynomial families, including
read-once arithmetic formulas, design polynomials, and the Hamiltonian cycle
polynomial
\cite{guptasahathankey2023soda,baraskardewansaha2024stacs,
dewan2026hamiltonian}. Baraskar, Dewan, Saha, and Sinha also proved NP-hardness of testing equivalence to sparse polynomials and to constant-support polynomials
under affine transformations
\cite{basarkardewansahasinahardness2024icalp}. Sparse SET is the translation-only, target-given special case, namely the case $A=I$ with both sparse endpoints $f$ and $g$ given as input. 

\paragraph{Shift equivalence testing.}
Grigoriev studied SET itself as a decision and search problem for densely
represented polynomials and gave deterministic, randomized, and quantum
algorithms in different characteristics \cite{grigorievtcs1997}. In
characteristic zero, the deterministic algorithm runs in time polynomial in the
dense representation. For $n$-variate polynomials of total degree at most $d$,
the dense representation has size $\binom{n+d}{d}$ and the running time is
$n^{O(d)}$. Padding a sparse input with zero coefficients therefore does not
give an efficient algorithm when the exponents are encoded in binary.

\paragraph{The result of \cite{dviroliveirashpilkai2014icalp} and its implication for sparse SET.}~Dvir, Oliveira, and Shpilka~\cite{dviroliveirashpilkai2014icalp} studied SET for arithmetic circuits, in both the white-box and black-box models \cite{dviroliveirashpilkai2014icalp}. They gave
a randomized algorithm running in time polynomial in the circuit size, the number of variables, and the degree, and showed that the randomness in their algorithm is used only for polynomial identity testing. Consequently,
derandomizing PIT would derandomize SET in their framework. They also obtained a general deterministic reduction for restricted circuit classes. In one form,
their result applies to classes $\mathcal{M}_1,\mathcal{M}_2$ such that
$\mathcal{M}_1$ is closed under taking homogeneous components and first-order directional derivatives, while $\mathcal{M}_2$ is closed under shifts. The
number of arithmetic operations is polynomial in $n,d$, and in the size of a hitting set, or of a white-box PIT algorithm, for the class of differences
$\mathcal{M}_1-\mathcal{M}_2$. They also pointed out an
analogy with algorithms for black-box groups, such as Sims' algorithm: their procedure can be viewed as finding the affine space of all shifts from $f$ to $g$, rather than merely deciding whether this space is empty
\cite{seress2003permutation}.

Although Section~6 of \cite{dviroliveirashpilkai2014icalp} observes that the
sparse class is not closed under shifts and gives no direct sparse-class
algorithm, a low-degree bound follows indirectly via ROABPs. A sparse polynomial
\[
        f(\bx)=\sum_{i=1}^s c_i\prod_{j=1}^n x_j^{\alpha_{i,j}}
\]
is a sum of products of univariate polynomials, and hence has a read-once
oblivious algebraic branching program of width at most $s$ in any variable
order. Choose the same variable order for both inputs. After a shift,
\[
        f(\bx+\ba)=\sum_{i=1}^s c_i\prod_{j=1}^n (x_j+a_j)^{\alpha_{i,j}},
\]
it still has such a representation of size $\poly(n,s,d)$. Therefore, the
general white-box theorem of \cite{dviroliveirashpilkai2014icalp}, together
with white-box ROABP PIT \cite{razshpilka2005}, yields a deterministic algorithm
for explicit sparse inputs using $\poly(n,s,t,d)$ arithmetic operations. This
is an indirect consequence rather than a sparse-case result stated in DOS'14. It is efficient in
the ordinary low-degree regime, but it is not efficient in the high-degree sparsity setting: the dependence is polynomial in $d$, whereas the sparse input length is polynomial in $\log d$. Moreover, this framework does not by itself address bit complexity over $\Qset$, the existence of a rational shift of polynomial height, or the problem of outputting such a small witness.

\section{Preliminaries}

\paragraph{Notation.}
\begin{enumerate}

\item \textbf{Monomials.} For vectors $\bz=(z_1,\ldots,z_n)$ and $\bv\in\Nset^n$, we write
\[
\bz^{\bv}=z_1^{v_1}\cdots z_n^{v_n}.
\]
When $\bz$ is a vector of variables, then $\bz^\bv$ is the monomial with exponent vector given by $\bv$. When $\bz$ is a vector of field elements, then it is the monomial evaluated at the given point.

\item \textbf{Degree and derivatives.} Throughout the paper, degree means total degree. We define $|\bv| := v_1+\cdots+v_n$, $\bv! := v_1!\cdots v_n!$, and
$\partial^{\bv}=\partial_{x_1}^{v_1}\cdots \partial_{x_n}^{v_n}$, where $\partial_{x_i}$ is the partial derivative taken with respect to the variable $x_i$.

\item \textbf{Partial order.} We use the coordinate-wise partial order $\preceq$ on $\Nset^n$: for
$\bv,\bu\in\Nset^n$, we write $\bv\preceq \bu$ if and only if $v_i\le u_i$
for every $i\in[n]$. A set $T\subseteq\Nset^n$ is called \emph{cone-closed} if,
whenever $\bu\in T$ and $\bv\preceq \bu$, we also have $\bv\in T$.

\item \textbf{Polynomial notation.} For a polynomial $f$, $H^\ell(f)$ denotes the homogeneous part of degree $\ell$ of $f$, that is, the sum of all monomials of $f$ that have degree $\ell$. For two polynomials $f,g\in\mathbb K[\bx]$, $S_{f,g}(\mathbb K)$ denotes the set of all shift vectors $\ba\in\mathbb K^n$ such that $f(\bx + \ba) = g(\bx)$. The set of all shifts $\ba\in\mathbb K^n$ that leave $f$ invariant, that is, $f(\bx + \ba) = f(\bx)$, is denoted by $S_f(\mathbb K)$. Over $\Qset$, we write $S_{f,g}$ for $S_{f,g}(\Qset)$. For a polynomial $f(\bx)$, $\sp(f)$ denotes the sparsity of $f$, that is, the number of nonzero monomials in $f$. For a polynomial $p$, we also write
\[
        H^{\ge r}(p)\coloneqq \sum_{\ell\ge r}H^\ell(p).
\]

\item \textbf{Height.} A rational number $r$ is said to have height at most $\tau$ if it can be written as $r=p/q$, where $p,q\in\Z$, $q\neq 0$, and $|p|,|q|\le 2^{\tau}$. This is denoted by $\Ht(r)\le \tau$. The height of a sparse polynomial, vector, or matrix is the maximum height of its coefficients or entries, and is denoted by $\Ht(\cdot)$.

\begin{remark}
In parts of the literature, this quantity is called logarithmic height. In this paper we simply call it height. Up to constant factors it is the bit-size of the numerators and denominators, so we use height and bit-size interchangeably in polynomial bounds.
\end{remark}

\end{enumerate}


\begin{definition}[Directional derivative]
    For $r\ge 1$, the order-$r$ directional derivative of
    $f(\bx)\in\Qset[\bx]$ in direction $\ba\in\Qset^n$ is
    \[
        f^{(r)}(\ba,\bx)
        =
        \sum_{\mathbf e\in [n]^r}
        \left(\prod_{i=1}^r a_{e_i}\right)
        \frac{\partial^r f}{\partial x_{e_1}\cdots\partial x_{e_r}}.
    \]
    We also define $f^{(0)}(\ba,\bx)=f(\bx)$. For each homogeneous component
    $H^\ell(f)$ of $f$, define
    \[
        f_\ell^{(r)}(\ba,\bx)
        =
        \sum_{\mathbf e\in [n]^r}
        \left(\prod_{i=1}^r a_{e_i}\right)
        \frac{\partial^r H^\ell(f)}{\partial x_{e_1}\cdots\partial x_{e_r}}.
    \]
\end{definition}

The following lemma, due to Dvir, Oliveira and Shpilka \cite{dviroliveirashpilkai2014icalp}, states that if the first directional derivatives are identical for two directions $\ba$ and $\bb$, then all higher order directional derivatives are identical, too.

\begin{lemma}[Derivative equality under shifts, \cite{dviroliveirashpilkai2014icalp}]\label{derivative-equality-all-orders}
    Let $f\in\Qset[\bx]$ and $\ba,\bb\in\Qset^n$. If
    $f^{(1)}(\ba,\bx) = f^{(1)}(\bb,\bx)$, then
    $f^{(r)}(\ba,\bx) = f^{(r)}(\bb,\bx)$ for all $r \geq 1$.
\end{lemma}

The next lemma allows us to write the homogeneous parts of a shifted polynomial in terms of their directional derivatives.

\begin{lemma}[Homogeneous Taylor expansion, \cite{dviroliveirashpilkai2014icalp}]\label{taylor-expansion-lemma}
    Let $\ba \in \Qset^n$ and let $f \in \Qset[\bx]$ have degree $d$.
    Then for $0\le i\le d$,
    \[
        H^i(f(\bx+\ba))
        =
        \sum_{j=i}^d\frac{1}{(j-i)!}f_j^{(j-i)}(\ba,\bx).
    \]
\end{lemma}

Finally, if $\ba$ is a shift that transforms $f$ into $g$, then $\bb$ transforms $f$ into $g$, too, if and only if the first directional derivatives are identical.

\begin{lemma}[Shift space via first derivatives, \cite{dviroliveirashpilkai2014icalp}]\label{relation-between-two-shifts}
    Let $f,g\in\Qset[\bx]$ and $\ba,\bb\in\Qset^n$.
    If $\ba \in S_{f,g}$, then $\bb \in S_{f,g}$ if and only if $f_i^{(1)}(\bb,\bx)=f_i^{(1)}(\ba,\bx)$ for $0 \leq i \leq d$.
\end{lemma}

For a polynomial $f(\bx)$ of degree $d$, a shift $\ba\in\Qset^n$, and
$0\leq i \leq d$, we define:
\[
P_{i,\ba}(f) \coloneqq H^i\paren{H^{\ge i+2}(f)(\bx+\ba)}.
\]
\begin{lemma}\label{lem:tail-truncated-taylor}
For any polynomial $f\in \Qset[\bx]$, we have
\[
 P_{i,\ba}(f) = H^i\paren{H^{\ge i+2}(f)(\bx+\ba)} =
        \sum_{j=i+2}^d\frac{1}{(j-i)!}f_j^{(j-i)}(\ba,\bx).
    \]
\end{lemma}

\begin{proof}

    We have
    \begin{align*}
        H^i\paren{H^{\ge i+2}(f)(\bx+\ba)}
        &=
        H^i\paren{\sum_{j=i+2}^d H^j(f)(\bx+\ba)}\\
        &=
        \sum_{j=i+2}^d H^i\paren{H^j(f)(\bx+\ba)}
    \end{align*}
Applying \cref{taylor-expansion-lemma} to each $H^j(f)$ gives an expansion of $H^i\paren{H^j(f)(\bx+\ba)}$ in terms of the directional derivatives of the homogeneous components of $H^j(f)$. Since $H^j(f) = f_j$ has only one homogeneous component, we get a single term in the expansion:
    $$ H^i\paren{H^j(f)(\bx+\ba)} = \frac{1}{(j-i)!}f_j^{(j-i)}(\ba,\bx).$$
    Plugging this into the previous equation proves the lemma.
\end{proof}

\subsection{Properties of \texorpdfstring{$S_{f,g}$}{Sfg}}
Recall that for $f,g\in\Qset[\bx]$,
$$
S_{f,g}(\Cset) = \{\ba \in \Cset^n : f(\bx+\ba)=g(\bx)\}
$$
where the equality is understood as an identity in $\Cset[\bx]$. We first write two elementary facts about this shift set and then use them for the rational bit-complexity analysis.

\begin{lemma}[{\cite[Observation 4.2, Lemma 4.5]{dviroliveirashpilkai2014icalp}}]\label{lem:S_f}
Let $f \in \Cset[\bx]$, and define
$$
S_f(\Cset) := \{\bu \in \Cset^n : f(\bx+\bu)=f(\bx)\}.
$$
Then $S_f(\Cset)$ is a linear subspace of $\Cset^n$, given by $S_f(\Cset) = \{\ba \in \Cset^n\mid \sum_{i=1}^n a_i\frac{\partial f}{\partial x_i} = 0\}$.
\end{lemma}

\begin{lemma}[{\cite[Lemma 4.4]{dviroliveirashpilkai2014icalp}}]
For $f,g \in \Cset[\bx]$, the set $S_{f,g}(\Cset)$ is either empty or an affine subspace of $\Cset^n$. More precisely, if $\ba_0 \in S_{f,g}(\Cset)$, then
$$
S_{f,g}(\Cset)=\ba_0+S_f(\Cset).
$$
\end{lemma}

\begin{proposition}[Existence of a shift in $\Qset^n$] \label{prop:rational:shift}
Let $f,g \in \Qset[\bx]$. If
$$
S_{f,g}(\Cset) := \{\ba \in \Cset^n : f(\bx+\ba)=g(\bx)\}
$$
is nonempty, then $S_{f,g}(\Qset)$ contains a rational point. Equivalently, there exists $\bq \in \Qset^n$ such that
$$
f(\bx+\bq)=g(\bx).
$$
\end{proposition}

\noindent The proof is deferred to \cref{app:missing-proofs}.

\section{Structural Sparsity Bounds}
We first prove the structural statement that drives the algorithm. The proof
needs to isolate one chosen shifted monomial from
$g(\bx)=f(\bx+\ba)$ while keeping control over what happens when the same
operator is applied to the sparse monomial expansion of $g$. The next lemma
provides exactly the indexing set needed for this. It gives an interpolation
basis for the exponent vectors of the monomials of $f$, but with the additional
cone-closed property: if a multi-index is selected, then every smaller
multi-index is selected as well. This closure is what later prevents the
differential operator from creating too many possible output monomials.

\subsection{Cone-closed basis}

\begin{lemma}[Cone-closed interpolation basis]\label{lem:cone-closed-basis}
    Let
    \[
        B=\{\bb_1,\ldots,\bb_s\}\subseteq \Nset^n
    \]
    be a finite set of exponent vectors. Assume that the elements of $B$ remain
    distinct as points of $\Qset^n$. For every $\bv\in\Nset^n$, define
    $M_\bv\in\Qset^s$ by
    \[
        (M_\bv)_j\coloneqq \bb_j^{\bv}
        \qquad(j\in[s]).
    \]
    Then there is a cone-closed set $T\subseteq\Nset^n$ such that $\{M_\bv:\bv\in T\}$ is a basis of $\Qset^s$. In particular, $|T|=s$.
\end{lemma}

\begin{proof}
    Let $\operatorname{ev}_B:\Qset[z_1,\ldots,z_n]\to\Qset^s$ be the evaluation map
    \[
        P\mapsto (P(\bb_1),\ldots,P(\bb_s)).
    \]
    We first check that this map is surjective. For each $m\in[s]$, let
    \[
        \mathfrak m_m=(z_1-b_{m,1},\ldots,z_n-b_{m,n})
        \]
    be the ideal of polynomials vanishing at $\bb_m$. Note that reducing a polynomial $P$ modulo
    $\mathfrak m_m$ is nothing but evaluating $P$ at the point $\bb_m$.
    
    The ideals $\mathfrak m_m$ are maximal, that is,
    they are not contained in any other proper ideal.
    Since the points $\bb_1,\ldots,\bb_s$ are pairwise distinct, the ideals are pairwise
    coprime. By the Chinese remainder theorem, the quotient map
    \[
        \Qset[z_1,\ldots,z_n]\to
        \prod_{m=1}^s \Qset[z_1,\ldots,z_n]/\mathfrak m_m
        \cong \Qset^s
    \]
    is surjective, and this quotient map is exactly $\operatorname{ev}_B$.
    The image of the monomial $\bz^{\bv}$ under this map is exactly $M_\bv$,
    so the family $\{M_\bv:\bv\in\Nset^n\}$ spans $\Qset^s$.

    Now fix the graded lexicographic order $\le_{\mathrm{grlex}}$ on
    $\Nset^n$: vectors are ordered first by total degree $|\bv|$ and then ties are
    broken lexicographically. This order is a well-order
    and refines $\preceq$, meaning that $\bv\preceq \bu$ implies
    $\bv\le_{\mathrm{grlex}}\bu$. Moreover, it is compatible with
    addition: if $\bv<_{\mathrm{grlex}}\bu$, then
    $\bv+\bw<_{\mathrm{grlex}}\bu+\bw$ for every $\bw\in\Nset^n$.

    We construct $T$ by choosing only dimension-increasing vectors. Start with
    $T_0=\emptyset$. If the current span is not all of $\Qset^s$, then choose
    $\bu_{m+1}$ to be the least vector, in graded lexicographic order, such that
    \[
        M_{\bu_{m+1}}\notin \vspan\{M_\bv:\bv\in T_m\},
    \]
    and set $T_{m+1}=T_m\cup\{\bu_{m+1}\}$. Such a vector exists whenever the
    current span is not all of $\Qset^s$, because the full family
    $\{M_\bv:\bv\in\Nset^n\}$ spans $\Qset^s$. Each step increases the dimension
    of the span by one, so after exactly $s$ steps the chosen vectors form a
    basis of $\Qset^s$. Let $T=T_s$.

    When $\bu_{m+1}$ is chosen, minimality implies that every
    $M_\bq$ with $\bq<_{\mathrm{grlex}}\bu_{m+1}$ already lies in
    $\vspan\{M_\bv:\bv\in T_m\}$. After adding $\bu_{m+1}$, the current span
    therefore contains every $M_\bq$ with
    $\bq\le_{\mathrm{grlex}}\bu_{m+1}$. In particular, no smaller vector can
    become rank-increasing later, so the selected vectors are increasing in
    graded lexicographic order.

    We now show that $T$ is cone-closed. Suppose $\bu\in T$, $\bv\preceq \bu$,
    and $\bv\notin T$. Since grlex refines $\preceq$, we have
    $\bv<_{\mathrm{grlex}}\bu$. Write $\bu=\bu_{m+1}$ (so $\bu$ is chosen at
    step $m+1$). Because $\bv$ occurs earlier in the grlex order and was not
    chosen when it was encountered, we have a dependence
    \[
        M_\bv=\sum_{\bw\in T_m,\ \bw<_{\mathrm{grlex}}\bv}\lambda_\bw M_\bw.
    \]
    Since $M_\br\odot M_\bq=M_{\br+\bq}$ for the Hadamard product $\odot$, which is defined as the pointwise product of the vectors,
    taking the Hadamard product with $M_{\bu-\bv}$ gives
    \[
        M_\bu
        =
        M_{\bu-\bv}\odot M_\bv
        =
        \sum_{\bw\in T_m,\ \bw<_{\mathrm{grlex}}\bv}\lambda_\bw M_{\bw+\bu-\bv}.
    \]
    Because graded lexicographic order is compatible with addition and
    $\bw<_{\mathrm{grlex}}\bv$, each index $\bw+\bu-\bv$ satisfies
    $\bw+\bu-\bv<_{\mathrm{grlex}}\bu$. By minimality of $\bu_{m+1}$, every
    such $M_{\bw+\bu-\bv}$ already lies in
    $\vspan\{M_\bq:\bq\in T_m\}$. Hence $M_\bu\in
    \vspan\{M_\bq:\bq\in T_m\}$, contradicting the choice of $\bu_{m+1}$ as
    rank-increasing. Therefore no such $\bv$ exists, and $T$ is cone-closed.
\end{proof}

\subsection{Structural lemma} \label{sec:struc-lem}

\begin{lemma}\label{lem:truncated-structural}
    Let $f,g\in\Qset[\bx]$ be polynomials of degree $d$ with sparsities $s$ and
    $t$, respectively.
    Let $\ba\in\Qset^n$, and suppose that
    \[
        H^{\ge r}(g(\bx))=H^{\ge r}(f(\bx+\ba)).
    \]
    Then, for every monomial $\bx^{\bgamma}$ appearing in $H^{\ge r}(f)$,
    \[
        \sp\left(H^{\ge r}\big((\bx+\ba)^{\bgamma}\big)\right)\le st.
    \]
\end{lemma}

\begin{proof}
    Write
    \[
        f(\bx)=\sum_{j=1}^s d_j\bx^{\balpha_j}
        \qquad\text{and}\qquad
        g(\bx)=\sum_{k=1}^t c_k\bx^{\bbeta_k},
    \]
    where all displayed coefficients are nonzero. Fix a monomial
    $\bx^{\bgamma}$ appearing in $H^{\ge r}(f)$, and relabel the monomials of
    $f$ so that $\balpha_1=\bgamma$.

    Since the exponent vectors are distinct integer vectors, the exponent
    vectors of the monomials of $f$ remain distinct as points of $\Qset^n$.
    Apply \cref{lem:cone-closed-basis} to these exponent vectors. We
    get a cone-closed set $T\subseteq\Nset^n$ with $|T|=s$ such that the
    vectors $M_\bv$, $\bv\in T$, defined by
    \[
        (M_\bv)_j=\balpha_j^{\bv}
        \qquad(j\in[s],\bv\in T)
    \]
    form a basis of $\Qset^s$. Hence there are coefficients $r_\bv$, $\bv\in T$,
    isolating the first coordinate:
    \begin{equation}\label{standard-basis-vector}
        \sum_{\bv \in T} r_\bv \balpha_j^{\bv} = \begin{cases}
        1 & \text{if } j = 1,\\
        0 & \text{otherwise}.
    \end{cases}
    \end{equation}
    Define the operators
    \[
        E_i = (x_i+a_i)\frac{\partial}{\partial x_i}
        \qquad\text{and}\qquad
        \bE^{\bv}=\prod_{i\in[n]}E_i^{v_i}.
    \]
    Since $E_i(\bx+\ba)^{\be}=e_i(\bx+\ba)^{\be}$, we have
    \[
        \bE^{\bv}(\bx+\ba)^{\be}=\be^{\bv}(\bx+\ba)^{\be}.
    \]
    Define
    \[
        \mathcal I_{\bgamma}=\sum_{\bv\in T} r_\bv \bE^{\bv}.
    \]
    The operator $\mathcal I_{\bgamma}$ never increases total degree: each
    $E_i=x_i\partial_{x_i}+a_i\partial_{x_i}$ preserves the degree or decreases
    it by one.
    Thus, for every polynomial $h$,
    \[
        H^{\ge r}(\mathcal I_{\bgamma}(h))=H^{\ge r}(\mathcal I_{\bgamma}(H^{\ge r}(h))).
    \]
    Using the assumed equality
    $H^{\ge r}(g)=H^{\ge r}(f(\bx+\ba))$, we get
    \[
        H^{\ge r}(\mathcal I_{\bgamma}(g))=H^{\ge r}(\mathcal I_{\bgamma}(f(\bx+\ba))).
    \]
    On shifted monomials, $\mathcal I_{\bgamma}$ acts diagonally. Indeed,
    $f(\bx+\ba)=\sum_{j=1}^s d_j(\bx+\ba)^{\balpha_j}$, so
    \begin{align*}
        \mathcal I_{\bgamma}(f(\bx+\ba))
        &=
        \sum_{j=1}^s d_j
        \left(\sum_{\bv\in T}r_\bv\balpha_j^{\bv}\right)
        (\bx+\ba)^{\balpha_j}\\
        &=
        d_1(\bx+\ba)^{\bgamma},
    \end{align*}
    where the last equality follows from \cref{standard-basis-vector}.
    Since $d_1\ne 0$, multiplication by $d_1$ does not change sparsity.
    Therefore,
    \[
        \sp\left(H^{\ge r}(\mathcal I_{\bgamma}(g))\right)
        =
        \sp\left(H^{\ge r}\big((\bx+\ba)^{\bgamma}\big)\right).
    \]

    It remains to bound $\sp(H^{\ge r}(\mathcal I_{\bgamma}(g)))$. For $\bv\in T$ and a monomial
    $\bx^{\bbeta_k}$ of $g$, applying $\bE^\bv$ can only change the exponent to
    $\bbeta_k-\bu$ for some $\bu\preceq\bv$, so its support is contained in
    monomials of the form
    \[
        \bx^{\bbeta_k-\bu}
        \qquad\text{with}\qquad
        \bu\preceq\bv,\ \bu\preceq\bbeta_k.
    \]
    Because $T$ is cone-closed, every such $\bu$ lies in $T$. Hence
    $H^{\ge r}(\mathcal I_{\bgamma}(g))$ is supported on monomials of the form
    \[
        \bx^{\bbeta_k-\bu}
        \qquad
        \text{where } \qquad k\in[t],\ \bu\in T.
    \]
    There are at most $t|T|=st$ such monomials. Therefore
    \[
        \sp\left(H^{\ge r}\big((\bx+\ba)^{\bgamma}\big)\right)\le st.
    \]
\end{proof}

Taking $r=0$ in \cref{lem:truncated-structural} proves the
$f$-monomial part of \cref{thm:intro-structural}. The same statement for
monomials of $g$ follows by applying the lemma to the identity
$f(\bx)=g(\bx-\ba)$.

\begin{proof}[Proof of the shifted-sparse PIT corollary, \cref{cor:intro-verify-shift}]
The identity
\[
        f(\bx+\ba)\equiv g(\bx+\bb)
\]
is equivalent, after the change of variables $\mathbf y=\bx+\bb$, to
\[
        f(\mathbf y+\bc)\equiv g(\mathbf y),
        \qquad\text{where }\bc=\ba-\bb.
\]
It is therefore enough to verify identities of the form
$f(\bx+\bc)\equiv g(\bx)$.

For each monomial $c_{\bgamma}\bx^{\bgamma}$ of $f$, compute
\[
        N_{\bgamma}=\prod_{i\in\supp(\bc)}(\gamma_i+1),
\]
but stop the multiplication once the value exceeds $st$. If
$N_{\bgamma}>st$ for some monomial, reject. This is sound: if the identity were
true, then \cref{thm:intro-structural} would imply $N_{\bgamma}\le st$ for every
monomial of $f$.

Otherwise every shifted monomial has at most $st$ terms. Enumerate all vectors
$\bv\preceq\bgamma$ with $\supp(\bv)\subseteq\supp(\bc)$, form the terms
\[
        c_{\bgamma}
        \prod_{i=1}^n \binom{\gamma_i}{v_i}(\bc_i)^{v_i}
        \bx^{\bgamma-\bv},
\]
collect like terms over all monomials of $f$, and compare the resulting sparse
polynomial with the sparse representation of $g$. The enumeration size is at
most $s\cdot st$, and all exponent arithmetic uses numbers of bit-length
$O(\log d)$, so the arithmetic operation count is
$\poly(n,s,t,\log d)$. Over $\Qset$, if the input coefficients and the shift
coordinates have polynomial bit-size, the same computation has polynomial bit
complexity, since every nonzero enumerated term has $|\bv|\le st-1$.
\end{proof}

\subsection{Height bound for computations in the algorithm} \label{sec:height-bound}
We now fix the single height cap used by the sparse tail computation and by
the main algorithm. In the height and bit-complexity statements below, every
polynomial bound denoted by $\poly(\cdot)$ is effective. Whenever the algorithm
uses such a bound as a height cap, one can replace the notation by a fixed
explicit polynomial in the displayed parameters, chosen large enough to
dominate the determinant bounds and sparse-expansion bounds proved below.
Fix once and for all a sufficiently large absolute constant $C$ and set
\[
    \mathcal B_{\mathrm{glob}}
    =
    \bigl(nst(\tau+1)\log(d+2)\bigr)^C .
\]
We choose $C$ large enough so that this single cap dominates all height bounds
used below, including the sparse-expansion bounds for the intermediate tails
and the determinant bounds for the basic rational solutions in the yes-instance
analysis. Thus every explicit height check in the algorithm can be implemented
by comparing against $\mathcal B_{\mathrm{glob}}$.

\begin{proposition}\label{prop:sparsity-and-height-of-P}
With $C$ chosen sufficiently large in the definition of
$\mathcal B_{\mathrm{glob}}$, the following holds.
Suppose $f,g\in\Qset[\bx]$ are polynomials of degree $d$ with
$\Ht(f), \Ht(g) \le \tau$ and $\sp(f)=s$, $\sp(g)=t$. Suppose
$g(\bx) = f(\bx+\bc)$ for a shift $\bc \in \Qset^n$ whose coordinates satisfy
the polynomial height bound from \cref{small-shift}. Then, for every
$0\leq i \leq d-1$,
\begin{align}
    \sp(P_{i,\bc}(f)) &\le s^2t, \\
    \Ht(P_{i,\bc}(f)) &\le \mathcal B_{\mathrm{glob}}.
\end{align}
\end{proposition}
\begin{proof}

Write
\[
        f(\bx)=\sum_{\ell=1}^s \lambda_\ell\bx^{\bgamma_\ell},
        \qquad \lambda_\ell\ne 0.
\]
By \cref{lem:tail-truncated-taylor},
\[
        P_{i,\bc}(f)
        =
        H^i\paren{H^{\ge i+2}(f)(\bx+\bc)}
        =
        H^i\paren{
            \sum_{\substack{\ell\in[s]\\ |\bgamma_\ell|\ge i+2}}
            \lambda_\ell(\bx+\bc)^{\bgamma_\ell}
        }.
\]
Since $g(\bx)=f(\bx+\bc)$, \cref{lem:truncated-structural} with
$r=0$ and $\ba=\bc$ implies that
\[
        \sp\left((\bx+\bc)^{\bgamma_\ell}\right)\le st
        \qquad\text{for every }\ell\in[s].
\]
Taking a homogeneous part cannot increase support. Therefore each input
monomial contributes at most $st$ monomials to $P_{i,\bc}$, and after collecting
like terms
\[
        \sp(P_{i,\bc}(f))\le s\cdot st=s^2t.
\]

It remains to bound coefficient heights. A nonzero term of
$(\bx+\bc)^{\bgamma_\ell}$ is indexed by a vector
$\bv\preceq\bgamma_\ell$: it is obtained by replacing $v_r$ copies of $x_r$ by
$c_r$ for each $r \in [n]$, and has coefficient
\[
        \lambda_\ell\prod_{r=1}^n \binom{\gamma_{\ell,r}}{v_r}c_r^{v_r}.
\]
If this term is nonzero, then $\supp(\bv)\subseteq\supp(\bc)$. For every
$\bu\preceq\bv$ along a chain from $\mathbf 0$ to $\bv$, the corresponding term
in $(\bx+\bc)^{\bgamma_\ell}$ is also nonzero and distinct. Hence the full
expansion has at least $|\bv|+1$ nonzero terms. Since its sparsity is at most
$st$ as noted above, we get
\[
        |\bv|\le st-1.
\]
Thus only binomial coefficients $\binom{\gamma_{\ell,r}}{v_r}$ with
$|\bv|\le st-1$ can occur in a nonzero term contributing to
$P_{i,\bc}(f)$. Their product has height $O(st\log d)$, and
$\prod_r c_r^{v_r}$ has polynomial height by
\cref{small-shift}. Multiplying by the input coefficient
$\lambda_\ell$ gives height
\[
        \poly(n,s,t,\tau,\log d)
\]
for each generated term. A coefficient of $P_{i,\bc}(f)$ is a sum of at most
$s$ such rational numbers, so its height is at most
$\poly(n,s,t,\tau,\log d)$. We can choose $C$ large enough so that this height is bounded by $\mathcal B_{\mathrm{glob}}$.
\end{proof}

\section{Main Algorithm and Proof of Correctness}

Before running the descent, handle the trivial cases as follows. If both
polynomials are constant, return $\mathbf 0$ when they are equal and return
\textsc{FAIL} otherwise. If exactly one polynomial is constant, or if the two
nonconstant polynomials have different total degrees, return \textsc{FAIL},
because translations preserve total degree.

The pseudocode below describes the rational bit-complexity version. For the
arithmetic version over a field $\mathbb K$ of characteristic zero or larger
than $d$, use the same descent, omit the global height cap and all height checks,
and choose any solution of each linear system over $\mathbb K$.

\begin{algorithm}[H]
\caption{Main Algorithm for Sparse Shift Equivalence}\label{alg:sparse-shift}
\begin{algorithmic}[1]
\Require Sparse representations of nonconstant polynomials
$f,g\in \Qset[x_1,\dots,x_n]$ with $\deg(f)=\deg(g)=d$, $\sp(f)=s$, $\sp(g)=t$, and height at most $\tau$.
\Ensure A shift in $S_{f,g}$, if one exists; otherwise \textsc{FAIL}

\State Set the global height cap
$\mathcal B_{\mathrm{glob}}=(nst(\tau+1)\log(d+2))^C$, for the fixed constant
$C$ chosen in \cref{sec:height-bound}.
\State Let $H^{\delta_1}(f),\ldots,H^{\delta_m}(f)$ be the nonzero
homogeneous components of $f$ of positive degree, where
$d=\delta_1>\cdots>\delta_m>0$. The degree-zero component, if present, is not
included.
\State Set $\ba_0\gets \mathbf 0$.
\State Set $\mathcal L_0$ to be the empty linear system in the variables
$\bb=(b_1,\ldots,b_n)$.
\For{$q=1,\ldots,m$}
    \State Compute the following polynomial using \cref{lem:compute-Pik} and
    the global cap $\mathcal B_{\mathrm{glob}}$; if the procedure returns \textsc{FAIL},
    return \textsc{FAIL}:
    \[
        P_q(\bx) \gets P_{\delta_q-1,\ba_{q-1}}(f)(\bx)
        =
        \sum_{j=\delta_q+1}^d
        \frac{1}{(j-\delta_q+1)!}\,
        f_j^{(j-\delta_q+1)}(\ba_{q-1},\bx).
    \]
    \Statex The subroutine returns \textsc{FAIL} if the collected $P_q$ has more than
    $s^2t$ monomials or a coefficient of height larger than
    $\mathcal B_{\mathrm{glob}}$.
    \State Let $\mathcal L_q$ be $\mathcal L_{q-1}$ together with the
    coefficient equations of
    \[
        H^{\delta_q-1}(g(\bx))
        =
        H^{\delta_q-1}(f(\bx))
        +f_{\delta_q}^{(1)}(\bb,\bx)
        +P_{q}(\bx).
    \]
    \State Solve the linear system $\mathcal L_q$ for $\bb$.
    \If{the above system has no solution}
        \State \Return \textsc{FAIL}
    \Else
        \State Let $\ba_q\gets \bb$, where $\bb$ is a basic rational solution of the system.
    \EndIf
\EndFor
\State Verify $f(\bx)-g(\bx-\ba_m)\equiv 0$ using the sparse verifier from the proof of \cref{cor:intro-verify-shift}.
\If{$f(\bx)-g(\bx-\ba_m)\equiv 0$}
    \State \Return $\ba_m$
\Else
    \State \Return \textsc{FAIL}
\EndIf
\end{algorithmic}
\end{algorithm}

\subsection{Computing the sparse representation of \texorpdfstring{$P_{i,\ba}(f)$}{Piaf}}

We now explain how to compute $P_q$ in \cref{alg:sparse-shift}.

\begin{lemma} \label{lem:compute-Pik}
Given $\ba_{q-1}$, the step that computes $P_q$ in
\cref{alg:sparse-shift} either returns \textnormal{\textsc{FAIL}} or computes the sparse
representation of
\[
        P_q(\bx)=P_{\delta_q-1,\ba_{q-1}}(f)(\bx)
\]
using $\poly(n,s,t,\log d)$ arithmetic operations. In the rational
bit-complexity model, the same procedure uses
\[
        \poly(n,s,t,\tau,\log d,\Ht(\ba_{q-1}))
\]
bit operations. If it succeeds, then
\[
        \sp(P_q)\le s^2t
        \qquad\text{and, over }\Qset,\qquad
        \Ht(P_q)\le \mathcal B_{\mathrm{glob}}.
\]
Moreover, if
\[
        H^{\ge\delta_q}(f(\bx+\ba_{q-1}))=
        H^{\ge\delta_q}(g(\bx)),
\]
then the enumeration checks in the procedure do not fail.
\end{lemma}
\begin{proof}
For every monomial $c_{\bgamma}\bx^{\bgamma}$ of
$H^{>\delta_q}(f)$, the terms of
$c_{\bgamma}(\bx+\ba_{q-1})^{\bgamma}$ of degree exactly $\delta_q$ are indexed
by vectors $\bv\preceq\bgamma$ satisfying
\begin{equation}\label{eq:degree-deltaq-vectors}
        \supp(\bv)\subseteq\supp(\ba_{q-1}),
        \qquad
        |\bv|=|\bgamma|-\delta_q.
\end{equation}
Here $v_\ell$ means that, from the factor
$(x_\ell+a_{q-1,\ell})^{\gamma_\ell}$, we choose $v_\ell$ copies of
$a_{q-1,\ell}$ and the remaining $\gamma_\ell-v_\ell$ copies of $x_\ell$.
Their coefficients are
\[
        c_{\bgamma}
        \prod_{\ell=1}^n
        \binom{\gamma_\ell}{v_\ell}(a_{q-1,\ell})^{v_\ell},
\]
and the corresponding monomial is $\bx^{\bgamma-\bv}$.

The only point at which we use the additional hypothesis in the lemma is to
show that these enumeration checks cannot reject. Under that hypothesis,
\cref{lem:truncated-structural} applies with $\ba=\ba_{q-1}$ and
$r=\delta_q$, and gives
\[
        \sp\left(
        H^{\ge\delta_q}\big((\bx+\ba_{q-1})^{\bgamma}\big)
        \right)
        \le st.
\]
The degree-$\delta_q$ terms are part of this polynomial. The vectors satisfying
\cref{eq:degree-deltaq-vectors} are exactly the choices that produce those
degree-$\delta_q$ monomials, and different vectors give different monomials.
Hence the displayed bound implies that there are at most $st$ such vectors.

We now give a procedure to enumerate all vectors satisfying \cref{eq:degree-deltaq-vectors}.
The procedure either returns a list of all such vectors or returns \textsc{FAIL} if there are more
than $st$ such vectors. The procedure is polynomial-time in the input size, which is polynomial
in $n,s,t,\tau,\log d,\Ht(\ba_{q-1})$. We first compute:
\[
        S\coloneqq\sum_{\ell:\,a_{q-1,\ell}\ne 0}\gamma_\ell.
\]

If $|\bgamma|-\delta_q>S$, then no vector satisfying
\cref{eq:degree-deltaq-vectors} exists, because the sum of the coordinates of
any such vector is $|\bgamma|-\delta_q$, while the sum of the coordinates where
$\ba_{q-1}$ is nonzero is only $S$. In this case the degree-$\delta_q$ part for
this monomial is empty. The degree-$(\delta_q-1)$ part is empty as well, since
it would require one additional copy of some coordinate of $\ba_{q-1}$. Thus
this monomial contributes nothing to $P_q$.

It remains to handle the case $|\bgamma|-\delta_q\le S$. If
$|\bgamma|-\delta_q\ge st$, the procedure returns \textsc{FAIL}. Under the
lemma's additional hypothesis, this case cannot occur. Indeed, since
$|\bgamma|-\delta_q\le S$, choose any $\bv\preceq\bgamma$ supported on
$\supp(\ba_{q-1})$ with $|\bv|=|\bgamma|-\delta_q$. For every
$\bu\preceq\bv$, the corresponding term in
$c_{\bgamma}(\bx+\ba_{q-1})^{\bgamma}$ is nonzero, distinct, and has degree at
least $\delta_q$. Therefore
\[
        \prod_{\ell:\,v_\ell>0}(v_\ell+1)
        \ge |\bv|+1
        =
        |\bgamma|-\delta_q+1
        \ge st+1,
\]
contradicting the displayed sparsity bound above.

It remains to handle the case $|\bgamma|-\delta_q\le st-1$ and
$|\bgamma|-\delta_q\le S$.
Let $\ell_1,\ldots,\ell_m$ be the indices for which $a_{q-1,\ell}\ne 0$. We
build a dynamic-programming table whose entries are short lists of vectors.
The entry $\mathcal V_i(h)$ records, with cap $st+1$, all ways to obtain total
drop $h$ using only the coordinates
$\ell_i,\ell_{i+1},\ldots,\ell_m$.
For $i=m+1,m,\ldots,1$ and $h=0,\ldots,|\bgamma|-\delta_q$, let
$\mathcal V_i(h)$ be a list of vectors
$(v_{\ell_i},\ldots,v_{\ell_m})$ with
\[
        0\le v_{\ell_r}\le \gamma_{\ell_r}
        \qquad\text{for }r=i,\ldots,m,
        \qquad
        v_{\ell_i}+\cdots+v_{\ell_m}=h.
\]
If there are more than $st$ such vectors, we keep only $st+1$ of them.

The table is filled from the end. We start with
\[
        \mathcal V_{m+1}(0)=\{()\},
        \qquad
        \mathcal V_{m+1}(h)=\emptyset\quad(h>0),
\]
where $()$ is the empty vector. For $i\le m$, the list $\mathcal V_i(h)$ is
obtained by trying
$e=0,\ldots,\min\{\gamma_{\ell_i},h\}$ and appending $e$ in front of every
vector in $\mathcal V_{i+1}(h-e)$:
\[
        (e,\bw)
        \qquad
        \text{where}\qquad
        \bw\in \mathcal V_{i+1}(h-e).
\]
We stop adding to an entry once it contains $st+1$ vectors. By induction on
$i$, each entry either contains all vectors satisfying its condition, or
contains $st+1$ such vectors. Hence
$\mathcal V_1(|\bgamma|-\delta_q)$ directly gives the vectors satisfying
\cref{eq:degree-deltaq-vectors}, after adding zero coordinates outside
$\{\ell_1,\ldots,\ell_m\}$.

If $\mathcal V_1(|\bgamma|-\delta_q)$ has size $st+1$, the procedure returns \textsc{FAIL}.
Otherwise it is the exact list of all degree-$\delta_q$ vectors. The table has
$O(nst)$ entries, each entry tries at most $st$ values of $e$ and stores at
most $st+1$ vectors, so this takes polynomial time. Each listed vector gives
the monomial and coefficient above.

The listed terms are exactly the degree-$\delta_q$ part
\[
        C_{\bgamma}(\bx)
        =
        H^{\delta_q}\left(c_{\bgamma}(\bx+\ba_{q-1})^{\bgamma}\right).
\]
We now recover the degree-$(\delta_q-1)$ contribution from
$C_{\bgamma}$ by a direct coefficient calculation. Write
\[
        D_{\ba_{q-1}}
        =
        \sum_{\ell=1}^n a_{q-1,\ell}\frac{\partial}{\partial x_\ell}.
\]
Let $\be_\ell$ denote the $\ell$th standard basis vector. Fix
$\bw\preceq\bgamma$ with $\supp(\bw)\subseteq\supp(\ba_{q-1})$ and
$|\bw|=|\bgamma|-\delta_q+1$. The monomial $\bx^{\bgamma-\bw}$ can arise in
$D_{\ba_{q-1}}(C_{\bgamma})$ only by differentiating, for some $\ell$ with
$w_\ell>0$, the degree-$\delta_q$ monomial corresponding to
$\bw-\be_\ell$. Thus its coefficient in $D_{\ba_{q-1}}(C_{\bgamma})$ is
\[
    c_{\bgamma}\ba_{q-1}^{\bw}
    \sum_{\ell:w_\ell>0}
    \binom{\bgamma}{\bw-\be_\ell}
    (\gamma_\ell-w_\ell+1).
\]
For every $\ell$ with $w_\ell>0$,
\[
    \binom{\bgamma}{\bw-\be_\ell}(\gamma_\ell-w_\ell+1)
    =
    w_\ell\binom{\bgamma}{\bw}.
\]
Hence the coefficient above equals
\[
    (|\bgamma|-\delta_q+1)c_{\bgamma}\binom{\bgamma}{\bw}
    \ba_{q-1}^{\bw},
\]
which is $(|\bgamma|-\delta_q+1)$ times the coefficient of $\bx^{\bgamma-\bw}$ in
$H^{\delta_q-1}(c_{\bgamma}(\bx+\ba_{q-1})^{\bgamma})$. Therefore the contribution
of $c_{\bgamma}\bx^{\bgamma}$ to $P_q$ is
\[
        \frac{1}{|\bgamma|-\delta_q+1}\,
        D_{\ba_{q-1}}\left(C_{\bgamma}(\bx)\right).
\]
Summing these contributions over all monomials of $H^{>\delta_q}(f)$ and
collecting like terms gives exactly
$P_{\delta_q-1,\ba_{q-1}}(f)$.

After the terms have been collected, we make one final size check. If the
result has more than $s^2t$ monomials, or if one of its coefficients has height
larger than the supplied height cap, then we return \textsc{FAIL}. In
\cref{alg:sparse-shift}, this cap is $\mathcal B_{\mathrm{glob}}$. Otherwise we
output the collected sparse representation.

The vector-list dynamic program, the binomial coefficients with
$|\bv|\le st-1$, the one additional multiplication by a coordinate of
$\ba_{q-1}$ and by an exponent in the directional derivative, the collection of
like terms, and the final height comparisons all use
\[
        \poly(n,s,t,\tau,\log d,\Ht(\ba_{q-1}))
\]
bit operations.
\end{proof}

\begin{remark}
The computation of $P_q$ uses the current vector $\ba_{q-1}$, but on a yes
instance its value can be bounded through the fixed witnessing shift $\bc$.
Indeed, the tail equality proved in \cref{thm:sparse-shift-correctness} gives
$P_q=P_{\delta_q-1,\bc}(f)$ on every yes instance.
Thus the sparsity and height of the computed tail are controlled by
\cref{prop:sparsity-and-height-of-P}, using the single fixed shift $\bc$ rather
than a bound that grows recursively with $q$.
\end{remark}

\begin{theorem}\label{thm:sparse-shift-correctness}
    \Cref{alg:sparse-shift} correctly determines whether there exists a shift in
    $S_{f,g}$.
\end{theorem}
\begin{proof}
First suppose $S_{f,g}=\emptyset$. Then for every vector $\ba$, the polynomial
$f(\bx)-g(\bx-\ba)$ is nonzero. Thus no candidate can pass the final
verification check. The algorithm may return \textsc{FAIL} earlier because an
enumeration exceeds the support cap, because a computed coefficient exceeds
$\mathcal B_{\mathrm{glob}}$, or because a linear system is infeasible. If
none of these early checks fails, then all intermediate sparse objects are
bounded by the explicit caps in the algorithm and the candidate coordinates are
bounded by the basic-solution determinant bound, so the final sparse verifier
from \cref{cor:intro-verify-shift} runs within the stated polynomial time and
rejects. Hence the algorithm can only return \textsc{FAIL} on no instances.

It remains to consider a yes instance. By \cref{small-shift}, fix
$\bc\in S_{f,g}$ of height $\poly(n,s,t,\tau,\log d)$. Then
$g(\bx)\equiv f(\bx+\bc)$. By
\cref{taylor-expansion-lemma}, for every $0\le i\le d-1$,
\begin{equation}\label{eq:true-shift-degree}
    H^i(g(\bx))
    =
    H^i(f(\bx))
    +
    f_{i+1}^{(1)}(\bc,\bx)
    +
    P_{i,\bc}(f)(\bx).
\end{equation}

We prove the following invariant at the beginning of iteration $q$:
\begin{enumerate}[label=(\roman*), leftmargin=*]
    \item Both $\bc$ and the current vector $\ba_{q-1}$ satisfy
    $\mathcal L_{q-1}$.
    \item Every solution $\bb$ of $\mathcal L_{q-1}$ satisfies
    \begin{equation}\label{eq:sparse-shift-invariant}
        f_{\delta_p}^{(1)}(\bb,\bx)
        \equiv
        f_{\delta_p}^{(1)}(\bc,\bx)
        \qquad\text{for all }1\le p<q.
    \end{equation}
    \item The current vector $\ba_{q-1}$ satisfies
    \begin{equation}\label{eq:high-truncation-invariant}
        H^{\ge \delta_q}(f(\bx+\ba_{q-1}))
        =
        H^{\ge \delta_q}(g(\bx)).
    \end{equation}
\end{enumerate}
For $q=1$, the first item holds because $\mathcal L_0$ is empty, the second
item is vacuous, and the third holds because $\ba_0=\mathbf 0$,
$\delta_1=d$, and translations do not change the top homogeneous component:
\[
        H^d(f)=H^d(f(\bx+\bc))=H^d(g).
\]

Now assume the invariant holds at the beginning of iteration $q$. The
high-truncation equality \cref{eq:high-truncation-invariant} is exactly the
hypothesis used in \cref{lem:compute-Pik} to apply
\cref{lem:truncated-structural} with $r=\delta_q$. Hence the enumeration checks
in the computation of $P_q$ do not fail.

We next compare the tail being computed with the corresponding tail at the true
shift $\bc$. If $j>\delta_q$ and $H^j(f)\ne 0$, then $j=\delta_p$ for some
$p<q$. Applying the induction invariant to the algorithm's current vector
$\ba_{q-1}$ gives equality of the first directional derivatives in the
directions $\ba_{q-1}$ and $\bc$. By
\cref{derivative-equality-all-orders},
\begin{equation}\label{eq:processed-all-derivatives}
    f_j^{(r)}(\ba_{q-1},\bx)
    \equiv
    f_j^{(r)}(\bc,\bx)
    \qquad\text{for every }r\ge 1 \quad\text{and }j> \delta_q.
\end{equation}
Therefore
\begin{equation}\label{eq:tail-sum-equality}
    P_q(\bx)
    =
    P_{\delta_q-1,\ba_{q-1}}(f)(\bx)
    =
    P_{\delta_q-1,\bc}(f)(\bx).
\end{equation}
By \cref{eq:tail-sum-equality}, the polynomial being checked is
$P_{\delta_q-1,\bc}(f)$, which depends only on the fixed witnessing shift
$\bc$. The structural lemma gives the required sparsity bound, so the size
checks in \cref{lem:compute-Pik} pass. In the rational bit-complexity version,
\cref{prop:sparsity-and-height-of-P} also applies: $P_{\delta_q-1,\bc}(f)$ has
height at most $\mathcal B_{\mathrm{glob}}$. Hence the height checks pass as
well, so the computation of $P_q$ does not return
\textsc{FAIL}.

Using \cref{eq:true-shift-degree,eq:tail-sum-equality} with
$i=\delta_q-1$, we get
\begin{equation}\label{eq:true-shift-linear-equation}
    H^{\delta_q-1}(g(\bx))
    =
    H^{\delta_q-1}(f(\bx))
    +f_{\delta_q}^{(1)}(\bc,\bx)
    +P_q(\bx).
\end{equation}
Thus $\bc$ satisfies the new equations added in iteration $q$, and hence
satisfies $\mathcal L_q$.

Let $\bb$ be any solution of $\mathcal L_q$. Since
$\mathcal L_{q-1}\subseteq\mathcal L_q$, the induction invariant gives
\cref{eq:sparse-shift-invariant} for all $p<q$. Subtracting
\cref{eq:true-shift-linear-equation} from the new equation in $\mathcal L_q$
gives
\begin{equation}\label{eq:new-derivative-equality}
    f_{\delta_q}^{(1)}(\bb,\bx)
    \equiv
    f_{\delta_q}^{(1)}(\bc,\bx).
\end{equation}
Therefore every solution of $\mathcal L_q$ satisfies the derivative equalities
for all $p\le q$.

The algorithm chooses $\ba_q$ to be a solution of $\mathcal L_q$. In the
rational bit-complexity version it chooses a basic rational solution; by
\cref{lem:basic-rational-solution-height-bit} in \cref{app:missing-proofs} and the choice of
$\mathcal B_{\mathrm{glob}}$, this solution has height at most
$\mathcal B_{\mathrm{glob}}$. If
$q<m$, we prove the high-truncation invariant for the next iteration. By the previous
paragraph and \cref{derivative-equality-all-orders}, all directional
derivative contributions from components $H^{\delta_p}(f)$ with $p\le q$ are
the same under $\ba_q$ and under $\bc$. For degrees at least
$\delta_{q+1}$, every contribution either comes from one of these already
processed higher components, or from the top part of the component of the same
degree, which is independent of the shift. Hence
\[
    H^{\ge \delta_{q+1}}(f(\bx+\ba_q))
    =
    H^{\ge \delta_{q+1}}(f(\bx+\bc))
    =
    H^{\ge \delta_{q+1}}(g(\bx)).
\]
This proves the invariant for the next iteration.

After the final iteration, the returned vector $\ba_m$ satisfies
\[
    f_{\delta_p}^{(1)}(\ba_m,\bx)
    \equiv
    f_{\delta_p}^{(1)}(\bc,\bx)
    \qquad\text{for every }1\le p\le m.
\]
If $H^j(f)=0$ or $j=0$, then
$f_j^{(1)}(\ba_m,\bx)=f_j^{(1)}(\bc,\bx)=0$. Therefore
\[
    f_j^{(1)}(\ba_m,\bx)\equiv f_j^{(1)}(\bc,\bx)
    \qquad\text{for all }0\le j\le d.
\]
Since $\bc\in S_{f,g}$, \cref{relation-between-two-shifts} implies
$\ba_m\in S_{f,g}$. The final verification succeeds, and the algorithm returns
a valid shift.
\end{proof}

\begin{proposition}[Runtime and height bounds]
\label{prop:intermediate-height}
In the rational bit-complexity version, every basic rational solution $\ba_q$
computed by \cref{alg:sparse-shift} has height
\[
    \poly(n,s,t,\tau,\log d).
\]
Moreover, the algorithm can be implemented using
$\poly(n,s,t,\tau,\log d)$ bit operations.
\end{proposition}

\begin{proof}
There are at most $s$ nonzero positive homogeneous components of $f$, so the
algorithm has at most $s$ iterations. In iteration $q$, only the degree
$\delta_q-1$ equations are newly added; all earlier equations are kept in
$\mathcal L_{q-1}$.

We first record the size and height of the equations added in one stage. Suppose
that the algorithm reaches the linear-system step in iteration $q$. Then the
computation of $P_q$ has been accepted by the global checks. Hence
\[
    \sp(P_q) \le s^2t,
    \qquad
    \Ht(P_q) \le \mathcal B_{\mathrm{glob}}.
\]
Therefore the right-hand side $H^{\delta_q-1}(g)-H^{\delta_q-1}(f)-P_q$ has support at most $t+s+s^2t$
and height bounded by $\poly(n,s,t,\tau,\log d)$.
Here we use the fact that $\mathcal B_{\mathrm{glob}}$ was chosen to be a fixed polynomial bound large enough to dominate all height bounds used by the algorithm. The newly added equations are obtained by comparing coefficients in
\[
    f_{\delta_q}^{(1)}(\bb,\bx)
    =
    H^{\delta_q-1}(g)-H^{\delta_q-1}(f)-P_q .
\]
The right-hand side has support at most $t+s+s^2t$. On the other hand,
\[
    f_{\delta_q}^{(1)}(\bb,\bx)
    =
    \sum_{\ell=1}^n b_\ell \partial_\ell H^{\delta_q}(f).
\]
Thus the support of the left-hand side is contained in the union of the supports
of the $n$ partial derivatives of $H^{\delta_q}(f)$. Since
$\sp(H^{\delta_q}(f))\le s$, this union has size at most $ns$. Consequently,
the coefficient comparison in iteration $q$ introduces at most $ns+t+s+s^2t =\poly(n,s,t)$ scalar linear equations.

The coefficients of these equations come from the partial derivatives
$\partial_\ell H^{\delta_q}(f)$. Differentiating a monomial only multiplies its
coefficient by an exponent bounded by $d$, so the entries of the coefficient
matrix have height $O(\tau+\log d)$. The right-hand side has height $\poly(n,s,t,\tau,\log d)$, as shown above.
Therefore, after any iteration $q$, the accumulated system $\mathcal L_q$ has
only $\poly(n,s,t)$ scalar equations in the $n$ variables $b_1,\ldots,b_n$, and
all entries of its coefficient matrix and right-hand side have height $\poly(n,s,t,\tau,\log d)$.

We now apply \cref{lem:basic-rational-solution-height-bit} in \cref{app:missing-proofs}. If the accumulated
system $\mathcal L_q$ is feasible, then a basic rational solution $\ba_q$ has
height $\poly(n,s,t,\tau,\log d)$,
and such a solution can be computed using $\poly(n,s,t,\tau,\log d)$ bit
operations. On yes instances, the choice of $\mathcal B_{\mathrm{glob}}$ is
large enough to dominate this polynomial height bound. On no instances, this
polynomial bound is sufficient for the bit-complexity analysis; no separate
height cap is imposed on $\ba_q$.

It remains to justify that the whole algorithm has polynomial bit complexity.
We argue inductively only for the cost of computing the successive tails and
solutions; the height of $P_q$ itself is controlled by the global check above,
not by a recursive height bound. The initial vector is
$\ba_0=\mathbf 0$, so it has height $0$. Suppose that the algorithm has reached
iteration $q$ and that $\ba_{q-1}$ has height
\[
    \poly(n,s,t,\tau,\log d).
\]
Then \cref{lem:compute-Pik} computes $P_q$, or returns \textsc{FAIL}, using $\poly(n,s,t,\tau,\log d)$
bit operations. If the computation succeeds, the global support and height
checks give
\[
    \sp(P_q)\le s^2t,
    \qquad
    \Ht(P_q)\le \mathcal B_{\mathrm{glob}},
\]
so the newly added equations have the polynomial size and height bounds proved
above. Hence exact linear algebra, together with the appendix 
\cref{lem:basic-rational-solution-height-bit}, computes the next basic rational
solution $\ba_q$ using
\[
    \poly(n,s,t,\tau,\log d)
\]
bit operations, and this solution again has height
\[
    \poly(n,s,t,\tau,\log d).
\]
This completes the induction for every candidate vector produced by the algorithm.

Since there are at most $s$ iterations, the total cost of all computations of
the polynomials $P_q$ and all linear-system solves is $\poly(n,s,t,\tau,\log d)$.
Finally, if the algorithm reaches the last verification step, the candidate shift has polynomial height. Therefore the final sparse verification can be
performed in $\poly(n,s,t,\tau,\log d)$ 
bit operations by the proof of \cref{cor:intro-verify-shift}. Hence the whole
rational bit-complexity version of the algorithm runs in $\poly(n,s,t,\tau,\log d)$
bit operations, and every basic rational solution $\ba_q$ computed by the algorithm has height $\poly(n,s,t,\tau,\log d)$.
\end{proof}

\section{Conclusion and Open Problems}

We gave the first deterministic polynomial-time algorithm, in the sparse input
size, for target-given shift equivalence of sparse multivariate polynomials over
$\Qset$. For $f,g\in \Qset[\bx]$ of sparsities $s,t$ and degree at most $d$,
the algorithm decides whether there is $\ba\in \Qset^n$ such that
\[
    f(\bx+\ba)\equiv g(\bx)
\]
using polynomially many arithmetic operations in $n,s,t,\log d$. It outputs a rational shift
of polynomial height and runs in polynomial bit complexity, with polynomial
bounds in $n,s,t,\tau,\log d$. As noted in the introduction, the
structural and arithmetic parts extend to any field of characteristic zero or
larger than $d$. This exponentially improves on the indirect ROABP-based
consequence of the framework of Dvir, Oliveira, and Shpilka described above
\cite{dviroliveirashpilkai2014icalp}, whose dependence on the degree is
polynomial in $d$ rather than polynomial in $\log d$.

The key point is that, although sparse polynomials are not closed under shifts, our structural lemma shows that the shifted expansions arising in a yes instance remain sparse after the relevant truncations. This allows the homogeneous
descent to be carried out entirely in sparse representation. Over $\Qset$, the same structural control also keeps the bit-complexity analysis from becoming recursive in the heights of the intermediate candidates.

In the algebraic-RAM model described above, separate black-box access to two
promised-sparse inputs is handled by sparse interpolation. This differs from receiving evaluation
access only to the composed expression $f(\bx+\ba)-g(\bx+\bb)$, for which our
structural verifier cannot inspect $f$ and $g$. Black-box PIT for any-order
ROABPs \cite{DBLP:journals/toc/GurjarKS17,guo2020improved} gives polynomial
dependence on $d$ and quasipolynomial-type dependence overall; obtaining
complexity polynomial in $n,s,t,\log d$ for this composed-oracle problem remains
open.

Another direction is to understand the small-characteristic regime. Our algorithm uses Taylor expansion, directional derivatives, and factorial denominators, and therefore assumes that the characteristic is zero or larger
than the degree. In small characteristic, several derivative-based characterizations no longer apply in the same form. It would be interesting to determine whether sparse shift
equivalence still admits an efficient deterministic algorithm for formal polynomial identities in small characteristic.

\paragraph*{Acknowledgements.}
The authors thank OpenAI's ChatGPT Pro for assistance during the preparation of
this manuscript. ChatGPT Pro also suggested the argument leading to the small rational shift existence proof of~\cref{small-shift}. The
authors have independently checked the arguments and take full responsibility
for the correctness of all results and proofs.

\bibliographystyle{alphaurl}
\bibliography{ref}

\appendix
\crefalias{section}{appendix}

\section{Missing Proofs}\label{app:missing-proofs}

\begin{observation}\label[observation]{obs:affine-average}
Let $A\subseteq \Cset^n$ be an affine subspace. If $\ba_1,\ldots,\ba_m\in A$ and $\lambda_1,\ldots,\lambda_m\in \Cset$ satisfy
$$
\sum_{i=1}^m \lambda_i=1,
$$
then
$$
\sum_{i=1}^m \lambda_i\ba_i\in A.
$$
In particular, affine subspaces are closed under affine averages.
\end{observation}

\begin{proof}
Write $A=\ba_0+V$, where $V\subseteq \Cset^n$ is a linear subspace. For each $i$, write $\ba_i=\ba_0+\bv_i$ with $\bv_i\in V$. Then
$$
\sum_{i=1}^m \lambda_i\ba_i =\ba_0+\sum_{i=1}^m \lambda_i\bv_i\in \ba_0+V=A,
$$
since $\sum_i\lambda_i=1$ and $V$ is closed under linear combinations.
\end{proof}

\begin{proof}[Proof of \cref{prop:rational:shift}]
Write
\[
        f(\bx+\bz)-g(\bx)=\sum_{\alpha} F_\alpha(\bz)\bx^\alpha,
\]
where each $F_\alpha \in \Qset[\bz]$. Thus $S_{f,g}(\Cset)$ is the common zero set of finitely many polynomials with rational coefficients:
\[
S_{f,g}(\Cset)=\{\ba \in \Cset^n : F_\alpha(\ba)=0 \text{ for all } \alpha\}\,.
\]
Since $S_{f,g}(\Cset)$ is nonempty, choose $\bc \in S_{f,g}(\Cset)$. Then $F_\alpha(\bc)=0$ for all $\alpha$. Now, consider the ideal
\[
I:=\langle F_\alpha\rangle_\alpha
\subseteq \overline{\Qset}[\bz]\;.
\]
We claim that $I$ is a proper ideal. Indeed, if $1\in I$, then we could write
\[
1=\sum_\alpha H_\alpha(\bz)F_\alpha(\bz)
\]
with $H_\alpha\in \overline{\Qset}[\bz]$, and evaluating at $\bz=\bc$ would give $1=0$, a contradiction. Thus $I$ is proper. Since $\overline{\Qset}$ is algebraically closed, the weak Nullstellensatz gives a point
\[
\ba\in \overline{\Qset}^{n}
\]
such that $F_\alpha(\ba)=0$ for all $\alpha$. Equivalently,
$f(\bx+\ba)=g(\bx)$.

Let $L/\Qset$ be a finite Galois extension containing all coordinates of $\ba$, and let
\[
G:=\operatorname{Gal}(L/\Qset).
\]
For every $\sigma \in G$, applying $\sigma$ to the identity $f(\bx+\ba)=g(\bx)$ gives
\[
f(\bx+\sigma(\ba))=g(\bx),
\]
because the coefficients of $f$ and $g$ lie in $\Qset$. Hence $\sigma(\ba) \in S_{f,g}(\Cset)$ for every $\sigma \in G$.

Since $S_{f,g}(\Cset)$ is an affine subspace, it is closed under affine averages; see~\cref{obs:affine-average}. Therefore
\[
\bar{\ba} := \frac{1}{|G|}\sum_{\sigma \in G}\sigma(\ba)\;\in\;S_{f,g}(\Cset)\;.
\]
On the other hand, $\bar{\ba}$ is fixed by every element of $G$. Formally, for any $\tau\in G$, the sum is just a reordering
\[
\sum_{\sigma\in G}\tau(\sigma(\ba))=\sum_{\sigma\in G}\sigma(\ba)\;.
\]
Thus each coordinate of $\bar{\ba}$ lies in $L^G=\Qset$.\footnote{Here $L^G:=\{z\in L:\sigma(z)=z\text{ for all }\sigma\in G\}$ is the fixed field of $G$. Since $L/\Qset$ is Galois and $G=\operatorname{Gal}(L/\Qset)$, we have $L^G=\Qset$.}
Therefore
\[
\bar{\ba}\in \Qset^n.
\]
Consequently there exists a rational vector $\bq:=\bar{\ba} \in \Qset^n$ such that
\[
f(\bx+\bq)=g(\bx),
\]
as desired.
\end{proof}

\begin{proof}[Proof of \cref{small-shift}]
Suppose $S_{f,g}(\overline{\Qset})\neq\emptyset$. By
\cref{prop:rational:shift}, choose a rational shift
    $\bc\in S_{f,g}(\Qset)$.

The proof has three steps. First, the structural bound confines all relevant
derivatives of $f$ to a polynomial-dimensional space $W$. Second, on $W$ we
represent the translation from $f$ to $g$ by a small-height unipotent operator
and take its finite logarithm to recover $D_{\bc}f$. Third, we solve
$D_{\bb}f=D_{\bc}f$ for a small basic rational solution $\bb$; the difference
$\bb-\bc$ stabilizes $f$, so $\bb$ is again a valid shift.

If $f$ is constant, then $g=f$, and the zero shift is a rational witness of
height $0$. Thus we may assume that $f$ is nonconstant.

Let $S=\supp(\bc)$ and
    \[
        W=\left\langle
            \partial^{\bv} f\ \middle|\ \bv\in\Nset^n,\ \supp(\bv)\subseteq S
        \right\rangle .
    \]
By \cref{lem:truncated-structural} with $r=0$, every monomial
    $\bx^{\bgamma}$ of $f$ satisfies
    \[
        \prod_{i\in S}(\gamma_i+1)\le st.
    \]

For a fixed monomial $\bx^{\bgamma}$ of $f$, a derivative
$\partial^{\bv}(\bx^{\bgamma})$ with $\supp(\bv)\subseteq S$ is
nonzero exactly when $0\le v_i\le\gamma_i$ for every $i\in S$. Therefore,
\begin{fact}\label{fact:deriv-num-ub}
    The
number of nonzero derivatives
\[
    \partial^{\bv}(\bx^{\bgamma}),
    \qquad \supp(\bv)\subseteq S,
\]
is at most
\[
    \prod_{i\in S}(\gamma_i+1)\le st.
\]
\end{fact}
\begin{fact}\label{fact:deriv-total-order}
   Each such nonzero derivative satisfies
\[
    |\bv|=\sum_{i\in S}v_i
    \le \sum_{i\in S}\gamma_i
    \le \prod_{i\in S}(\gamma_i+1)-1
    \le st-1.
\]
\end{fact}

Since $f$ has at most $s$ monomials, \cref{fact:deriv-num-ub} implies
\[
    \dim W\le s^2t.
\]
Let
\[
    m=\dim W\le s^2t.
\]
Choose a basis
\[
    B=\{w_1=f,w_2,\ldots,w_m\}
\]
of $W$ such that each
\[
    w_j=\partial^{\bv_j}f
\]
for some $\bv_j\in\Nset^n$ with $\supp(\bv_j)\subseteq S$. Since we take
$w_1=f$, we have $\bv_1=0$.

Now solve the linear system
\[
    \sum_{j=2}^m y_jw_j=g-f,
\]
in variables $y_2,\ldots,y_m$, where the linear equations of the system are obtained by comparing coefficients of monomials on both sides.

We claim that the system is satisfiable over $\Qset$. This is because
\[
    g(\bx)-f(\bx)
    =
    f(\bx+\bc)-f(\bx)
    =
    \sum_{\substack{\bv\ne 0\\ \supp(\bv)\subseteq S}}
        \frac{\bc^{\bv}}{\bv!}\partial^{\bv} f \in \operatorname{span}\{w_2,\ldots,w_m\},
\]
where the inclusion follows because $w_1=f$ is the only polynomial of degree $d$
in the basis $B$, while the displayed sum has degree at most $d-1$ and hence
lies in the span of $B\setminus \{w_1\}$.

By \cref{fact:deriv-total-order}, every coefficient of the derivatives $w_j = \partial^{\bv_j} f$ has height
at most
\[
    O(\tau+st\log d),
\]
because differentiating multiplies the input coefficients by falling factorials of
exponents and each exponent is bounded by $d$.
Therefore, the linear system has coefficient height
$O(\tau+st\log d)$ and it has $m-1 \leq s^2t$ variables. Then, by
\cref{lem:basic-rational-solution-height-bit}, it has a rational
solution $(y_2,\ldots,y_m)$ with height
\[
    \poly(s,t,\log d,\tau).
\]

Define the differential operator:
\begin{equation} \label{differential-operator}
Y=1+N \qquad \text{with} \qquad N=\sum_{j=2}^m y_j\partial^{\bv_j}.
\end{equation}
Then we have
\[
    Yf=f+\sum_{j=2}^m y_jw_j=g.
\]

\Cref{fact:deriv-total-order} says that any derivative of $f$ with respect to $S$-variables and total order $\geq st$ is zero. Therefore,
\[
    N^{st}h=0
    \qquad\text{for every }h\in W.
\]
Thus the logarithm of $Y$ is finite on $W$:
\[
    \log Y
    =
    \sum_{r=1}^{st-1}\frac{(-1)^{r+1}}{r}N^r.
\]
Put
\[
    q=(\log Y)f.
\]

The polynomial $q$ has coefficient height $\poly(n,s,t,\log d,\tau)$. Indeed,
work in the basis $B$ of $W$. Choose $m$ monomial coordinates in which the
coefficient matrix of $w_1,\ldots,w_m$ is nonsingular; this matrix has height
$O(\tau+st\log d)$. For each basis element $w_j$, the polynomial $Nw_j$ is a
polynomial-height linear combination of derivatives of $f$ of total order at
most $st-1$. Expressing $Nw_j$ back in $B$ by
\cref{lem:basic-rational-solution-height-bit} gives polynomial-height coordinates.
Thus the matrix $A$ of the transformation $N|_W$ has polynomial height. Since $A^{st}=0$,
\[
    q=\left(\sum_{r=1}^{st-1}\frac{(-1)^{r+1}}{r}A^r\right)e_1
\]
in $B$-coordinates, and repeated matrix multiplication and addition for $st$
steps keep the bit-height polynomial.

Now compare $Y$ with the shift operator
\[
    T_\bc
    =
    \exp\left(\sum_{i\in S}c_i\partial_i\right).
\]
We have
\[
    Yf=g=T_\bc f.
\]
Since $Y$ and $T_\bc$ are constant-coefficient differential operators, they commute with every $\partial^{\bv}$. Therefore, for every $\bv$ with $\supp(\bv)\subseteq S$,
\[
    (Y-T_\bc)\partial^{\bv} f
    =
    \partial^{\bv}(Y-T_\bc)f
    =
    0.
\]
Thus $Y$ and $T_\bc$ act identically on $W$. Their logarithms also act
identically on $W$, because the logarithm is a finite polynomial in the
nilpotent operator $Y-1=T_\bc-1$ on $W$. Hence
\[
    q
    =
    (\log Y)f
    =
    (\log T_\bc)f
    =
    \left(\sum_{i\in S}c_i\partial_i\right)f.
\]

Finally, solve the linear system in variables $b_i$, $i\in S$:
\[
    \sum_{i\in S}b_i\partial_i f=q.
\]
It is consistent, since $b_i=c_i$ is a solution. The coefficient matrix has
height $O(\tau+\log d)$, and the right-hand side $q$ has polynomial height.
Thus \cref{lem:basic-rational-solution-height-bit} gives a rational solution
$(b_i)_{i\in S}$ with
\[
    \Ht(b_i)\le \poly(n,s,t,\log d,\tau).
\]
Set $b_i=0$ for $i\notin S$. Then
\[
    \sum_{i=1}^n b_i\partial_i f
    =
    \sum_{i=1}^n c_i\partial_i f,
\]
so
\[
    \sum_{i=1}^n (b_i-c_i)\partial_i f
    = 0
\]
Hence, by \cref{lem:S_f}, we have $\bb-\bc \in S_f(\Cset)$. It follows that $\bb\in S_{f,g}\cap\Qset^n$ and has the desired height bound.
\end{proof}

\begin{definition}[Basic rational solution] \label{def:basic-rational-solution}
Let $M\mathbf y=\mathbf r$ be a feasible linear system over $\Qset$, with
$M\in\Qset^{m\times n}$ and $\rho=\rank(M)$. Choose $\rho$ linearly independent
columns of $M$ whose span contains $\mathbf r$, set all other variables to zero,
and solve for the variables in the chosen columns. The resulting vector is
called a \emph{basic rational solution}.
\end{definition}

\begin{lemma}[Basic rational solution: height and bit complexity]
\label{lem:basic-rational-solution-height-bit}
Let $M\in\Qset^{m\times n}$ and $\mathbf r\in\Qset^m$ be a linear system
$M\mathbf y=\mathbf r$ such that every entry of $M$ and $\mathbf r$ has
height at most $H$. If the system is feasible, then:
\begin{enumerate}[label=(\roman*), leftmargin=*]
    \item there exists a basic rational solution $\mathbf y^\star\in\Qset^n$,
    and
    \item every coordinate of $\mathbf y^\star$ has height
    $\poly(m,n,H)$.
\end{enumerate}
Moreover, such a basic rational solution can be computed deterministically in
$\poly(m,n,H)$ bit operations.
\end{lemma}

\begin{proof}
Let $\rho=\rank(M)$. Since the system is feasible, $\mathbf r$ lies in the
column span of $M$. Choose $\rho$ independent columns $B$ that form a basis of
this column span. Then $\mathbf r\in\vspan(M_B)$, so the reduced system
$M_B\mathbf y_B=\mathbf r$ is feasible. Choose $\rho$ independent rows $R$ such
that the square matrix $C=(M_B)_R$ is nonsingular, and let
$\mathbf d=\mathbf r_R$. Solving
\[
        C\mathbf z=\mathbf d
\]
determines the unique vector $\mathbf z$ with $M_B\mathbf z=\mathbf r$; the
remaining rows are satisfied because they are linear combinations of the rows
in $R$ on the column space spanned by $M_B$. Setting all variables outside $B$
to zero gives a basic rational solution of the original system.

For height, write each coordinate of $\mathbf z$ by Cramer's rule as
$\det(C_i)/\det(C)$, where $C_i$ denotes the matrix obtained from $C$ by
replacing its $i$th column with the right-hand side vector $\mathbf d$. By
Hadamard's bound, determinants of $r\times r$ matrices with entry height
$\le H$ have height $\poly(r,H)$, hence each coordinate has height
$\poly(m,n,H)$.

For computation, use fraction-free Gaussian elimination (or equivalent exact
linear algebra over $\Qset$) to find $\rho$, a nonsingular basis matrix $C$, and
the corresponding solution exactly; the bit-lengths of all intermediate integers
and the total bit-operation cost are bounded by $\poly(m,n,H)$.
\end{proof}

\end{document}